\RequirePackage{amsmath} 
\documentclass[runningheads]{llncs}
\pdfoutput=1

\makeatletter%
\@ifclassloaded{IEEEtran}%
  {\IEEEoverridecommandlockouts}%
  {}%
\makeatother%

\usepackage[dvipsnames]{xcolor}
\usepackage{comment} 
\usepackage{bussproofs}
\usepackage{mathtools}
\usepackage{amssymb}
\usepackage{xifthen}
\usepackage{tikz}
\usepackage{cancel}
\usepackage{enumerate}
\usepackage{placeins} 
\usepackage{hyperref} 
\usetikzlibrary{decorations.pathmorphing} 
\usetikzlibrary{arrows,automata,calc,positioning,arrows.meta,shapes,matrix}
\usepackage{graphicx} 

\newcommand{\isdef}[1][0mm]{\mathrel{\hspace{#1}=_\text{}\hspace{#1}}}
\newcommand{\ifdef}[1][0mm]{\mathrel{\hspace{#1}\iff_\text{}\hspace{#1}}}
\newcommand{\ifdefsmall}[1][0mm]{\mathrel{\hspace{#1}\Leftrightarrow_\text{}\hspace{#1}}}
\newcommand{\lts}[1][]{
 \ifthenelse{\equal{#1}{}}
 {\mathcal{IA}}
 {\mathcal{IA}(#1)}
}
\newcommand{\after}[3][]{#2 \mathrel{\operatorname{after}_{#1}} #3}

\newcommand{\inp}{\operatorname{in}}
\newcommand{\out}{\operatorname{out}}
\renewcommand{\det}{\operatorname{det}}
\newcommand{\detiu}{\operatorname{det}_{\operatorname{\textit{iu}}}}
\newcommand{\traces}{\operatorname{traces}}
\newcommand{\irtraces}{\operatorname{Ftraces}}
\newcommand{\iutraces}{\operatorname{IU}}
\newcommand{\oetraces}{\operatorname{OE}}
\newcommand{\utraces}{\operatorname{Utraces}}
\newcommand{\ioco}[1][]{\mathrel{\operatorname{\bf ioco_{\text{$#1$}}}}}
\newcommand{\uioco}{\mathrel{\operatorname{\bf uioco}}}

\newcommand{\iots}{\mathcal{IOTS}}
\newcommand{\ttag}[1]{\tag*{\text{[#1]}}}
\newcommand{\qedtag}{\tag*{\qed}}
\newcommand{\aaee}{\le_{\forall\forall\exists\exists}^\textit{tb}}

\newcommand{\ati}{\le_\textit{atc}}

\newcommand{\ir}{\le_\textit{if}}
\newcommand{\irrev}{\ge_\textit{if}}
\newcommand{\ireq}{\equiv_\textit{if}}
\newcommand{\iuoe}{\le_\textit{iuoe}}

\newcommand{\rcl}{\operatorname{fcl}}
\newcommand{\as}{\le_{as}}
\newcommand{\notrel}[2][3pt]{\mathrel{\hspace{#1}\cancel{\hspace{-#1}#2\hspace{-#1}}\hspace{#1}}}
\newcommand{\trans}{\operatorname{trans}}
\newcommand{\strat}{f}
\newcommand{\stratsys}{\strat_\textit{o}}
\newcommand{\stratenv}{\strat_\textit{i,tb}}
\newcommand{\stratenvati}{\strat_\textit{i}}
\newcommand{\Stratgen}{\Sigma}
\newcommand{\Stratsys}{\Stratgen_\textit{o}}
\newcommand{\Stratenv}{\Stratgen_\textit{i,tb}}
\newcommand{\Stratenvati}{\Stratgen_\textit{i}}
\newcommand{\stratdet}{\strat_\textit{d}}
\newcommand{\Stratdet}{\Stratgen_\textit{d}}
\newcommand{\stratrc}{\strat_\textit{r}}
\newcommand{\Stratrc}{\Stratgen_\textit{r}}
\newcommand{\last}{\operatorname{last}}
\newcommand{\iftracesdom}{\mathcal{FT}}

\makeatletter%
\@ifclassloaded{llncs}%
  {}%
  {%
    \newtheorem{definition}{Definition}[section]%
    \newtheorem{lemma}{Lemma}[section]%
    \newtheorem{theorem}{Theorem}[section]%
    \newtheorem{proposition}{Proposition}[section]%
    \newtheorem{corollary}{Corollary}[section]%
    \newtheorem{example}{Example}[section]%
    \newenvironment{proof}{\paragraph{Proof}}{\hfill$\square$\\}%
  }%
\makeatother%

\newtheorem{fact}{Fact}

\newcommand{\Path}{\operatorname{paths}}
\newcommand{\trace}{\operatorname{trace}}
\newcommand{\Outc}[1][]{
  \ifthenelse{\equal{#1}{}}
    {\mathop{\operatorname{outc}}}
    {\mathop{\operatorname{outc}}(\stratenv^{#1},\stratsys^{#1},\stratdet^{#1},\stratrc^{#1})}
}
\newcommand{\Outcati}[1][]{
  \ifthenelse{\equal{#1}{}}
    {\mathop{\operatorname{outc}}}
    {\mathop{\operatorname{outc}}(\stratenvati^{#1},\stratsys^{#1},\stratdet^{#1},\stratrc^{#1})}
}

\definecolor{myRed}{rgb}{0.7,0,0}
\definecolor{myGreen}{rgb}{0,0.5,0}

\tikzstyle{named}=[initial text=,inner sep=0.5pt,shape=circle,draw=black]
\tikzstyle{namedlong}=[initial text=, inner sep=-2pt, minimum width=20pt, minimum height=20pt, shape=ellipse, draw=black]
\tikzstyle{every state}=[initial text=,minimum size=0.15cm,inner sep=0pt,fill=black]
\tikzset{trans/.style={>={Latex[length=2mm]},-{Latex[length=2mm]}}}
\tikzset{every initial by arrow/.append style={trans}}

\newcommand{\tikzlts}{
 \tikzset{every loop/.append style={trans}}
 \tikzset{every edge/.append style={trans}}
}

\newif\ifhideproofs
\newif\ifshowshortversion

\newenvironment{longversion}{\ifvmode\else\unskip\fi}{\ignorespacesafterend}
\newenvironment{shortversion}{\ifvmode\else\unskip\fi}{\ignorespacesafterend}
\ifshowshortversion
  \excludecomment{longversion}
\else
  \excludecomment{shortversion}
\fi
  
\ifhideproofs
  \usepackage{environ}
  \NewEnviron{hidep}{}

\fi

\ifshowshortversion
  \title{Relating Alternating Relations for\\ Conformance and Refinement\thanks{Funded by the Netherlands Organisation of Scientific Research (NWO-TTW), project 13859: SUMBAT - SUpersizing Model-BAsed Testing}}%
\else
  \title{Relating Alternating Relations for\\ Conformance and Refinement\thanks{Funded by the Netherlands Organisation of Scientific Research (NWO-TTW), project 13859: SUMBAT - SUpersizing Model-BAsed Testing}}%
\fi
\author{Ramon Janssen\textsuperscript{1} \and Frits Vaandrager\textsuperscript{1} \and Jan Tretmans\textsuperscript{1 2}}%
\institute{Radboud University, Nijmegen, The Netherlands \email{\{ramonjanssen,f.vaandrager,tretmans\}@cs.ru.nl} \and
ESI (TNO), Eindhoven, The Netherlands}%
\authorrunning{R. Janssen, F. Vaandrager, J. Tretmans}%
\begin{longversion}
\subtitle{Technical Report}
\end{longversion}
\titlerunning{Relating Alternating Relations for Conformance and Refinement}

\begin{document}

\maketitle

\begin{abstract}

Various relations have been defined to express refinement and conformance
for state-transition systems with inputs and outputs,
such as $\ioco$ and $\uioco$ in the area of model-based testing,
and \emph{alternating simulation} and \emph{alternating-trace containment}
originating from game theory and formal verification.
Several papers have compared these independently developed relations, 
but these comparisons make assumptions (\emph{e.g.,} input-enabledness),
pose restrictions (\emph{e.g.,} determinism -- then they all coincide),
use different models (\emph{e.g.,} interface automata and Kripke structures),
or do not deal with the concept of \emph{quiescence}.
In this paper, we present the integration of the $\ioco$/$\uioco$ theory
of model-based testing and the theory of alternating refinements, within
the domain of non-deterministic, non-input-enabled interface automata.
A standing conjecture is that $\ioco$ and alternating-trace containment coincide.
Our main result is that this conjecture does not hold, but that $\uioco$ coincides with a variant of alternating-trace containment, for image finite interface automata
and with explicit treatment of quiescence.
From the comparison between $\ioco$ theory and alternating refinements, we conclude that $\ioco$ and the original relation of alternating-trace containment are too strong for realistic black-box scenarios.
We present a refinement relation which can express both $\uioco$ and refinement in game theory, while being simpler and having a clearer observational interpretation.



\end{abstract}


\section{Introduction}

Many software systems can be modelled using some kind of
state-transition automaton.
States in the model represent an abstraction of the states of the system,
and transitions between states model the actions that the system
may perform.
Depending on the kind of state-transition model,
an action can be the acceptance of an input,
the production of an output, an internal computation of the system,
the combination of an input and corresponding output,
or just an abstract, uninterpreted 'action' of the system.
Formal relations between state machines are often used to express
some notion of refinement, implementation correctness, or conformance:
$s_1$ is related to $s_2$ expresses that $s_1$ implements, refines,
or conforms to $s_2$.
Many such relations have been defined over the years, expressing
different intuitions of what constitutes a conforming implementation
or a correct refinement.

In this paper, we focus on state-transition systems
where actions are interpreted as either input or output.
An input involves a trigger from the environment to the system,
where the initiative is taken by the environment, whereas an output is initiated by
the system itself.
Modelling formalisms with inputs and outputs are, e.g.,
Input/Output Automata \cite{LyTu88},
Input-Output Transition Systems \cite{Tre96}, and
Interface Automata \cite{AlHe01}.
We use the latter in this paper.
%
We will extensively compare
the relations $\ioco$ and $\uioco$ from the area of model-based testing,
and \emph{alternating simulation} and \emph{alternating-trace containment}
originating from game theory and formal verification.
Previous papers have compared these independently developed relations,
but these comparisons make assumptions (\emph{e.g.,} input-enabledness),
pose restrictions (\emph{e.g.,} determinism -- then they all coincide),
use different models (\emph{e.g.,}
interface automata and Kripke structures),
or do not deal with the concept of \emph{quiescence},
i.e., the absence of outputs in a state,
that is crucial in the relations $\ioco$ and $\uioco$.
Based on this comparison, we propose the novel relation of \emph{input-failure refinement}, which links $\uioco$ and alternating-trace-containment.


\paragraph{$\ioco$.}

Model-based testing (MBT) is a form of black-box testing where a
System Under Test (SUT) is tested for conformance to a model.
The model is the basis for the algorithmic generation
of test cases and for the evaluation of test results.
Conformance is defined with a formal \emph{conformance} or
\emph{implementation relation} between SUTs and models.
\begin{longversion}
  Although an SUT is a black box, we can assume it could be modelled by
  some model instance in a domain of implementation models.
  This assumption is commonly referred to as the
  \emph{testability hypothesis} \cite{Gau95}, and it allows
  to reason about SUTs as if they were formal models.
\end{longversion}
\begin{shortversion}
  Although an SUT is a black box, we assume it could be modelled by
  some model instance in a domain of implementation models,
  so that we can reason about SUTs as if they were formal models.
\end{shortversion}

An often used conformance relation is $\ioco$
(\textbf{i}nput-\textbf{o}utput-\textbf{co}nformance) \cite{Tre96,Tre08}.
\begin{longversion}
The relation $\ioco$ is based on the testability hypothesis that implementations can be modelled as \emph{input-enabled} interface automata,
i.e., all states have a transition for all inputs.
\end{longversion}
\begin{shortversion}
The relation $\ioco$ is based on the assumption that implementations can be modelled as \emph{input-enabled} interface automata,
i.e., all states have a transition for all inputs.
\end{shortversion} An implementation $\ioco$-conforms to its specification
if the implementation never produces an output
that cannot be produced by the specification model in the same situation.
A particular, virtual output is \emph{quiescence},
actually expressing the absence of real outputs:
an implementation may only refuse outputs if the specification can do so.
Observing quiescence during a practical test is done by waiting for
a time-out.
The $\ioco$-testing theory has found its way into many MBT tools
and practical applications \cite{Mosetal09,Tre17}.

\begin{shortversion}
Whereas input-enabledness seems reasonable
for real-world systems, it is an inconvenience
in mathematical reasoning about $\ioco$,
in comparing specification models, and in stepwise refinement,
since the different domains for implementations and specifications
make that $\ioco$ is not reflexive and not transitive \cite{JaTr19}.
\end{shortversion}
\begin{longversion}
Whereas the testability hypothesis of input-enabledness may seem reasonable
for real-world software systems, it is an inconvenience
in mathematical reasoning about $\ioco$,
in comparing specification models, and in stepwise refinement,
since the different domains for implementations and specifications
make that $\ioco$ is not reflexive and not transitive \cite{JaTr19}.  
\end{longversion}

A variation of $\ioco$ is $\uioco$ \cite{BiReTr04}.
This relation is weaker than $\ioco$ and it was shown to have some
beneficial properties with respect to intuition of what conformance means,
as well as for formal reasoning about composition, transitivity, and
refinement \cite{BiReTr04,JaTr19}.
Moreover, a generalization of $\uioco$ was given in \cite{VoTr13}
that also applies to non-input-enabled implementations
and that is reflexive and transitive,
but a complete testing theory including test generation, test execution,
and test observations is still missing for this generalization,


\paragraph{Alternating Refinement.}

Originating from game theory, \emph{alternating refinement relations}
describe refinement as a game \cite{Aluetal98}.
Originally, alternating refinement was defined on alternating transition
systems, a variant of Kripke structures, which have state propositions
instead of input and output labels on transitions.
Behaviour of alternating transition systems is determined by \emph{agents}, which are either adversarial or collaborative.
\begin{shortversion}
In a two-player game of alternating refinement on two models $s_1$ and $s_2$, the antagonist controls the collaborative agents of $s_1$ and the adversarial agents of $s_2$.
The protagonist controls the collaborative agents of $s_2$ and the adversarial agents of $s_1$, and tries to ensure that the moves made by the agents in both models match.
\end{shortversion}
\begin{longversion}
The two-player game of alternating refinement on two models $s_1$ and
$s_2$ is then, in general, as follows.
The antagonist first chooses a move for the collaborative agents in $s_1$.
Second, the protagonist chooses a matching move for the collaborative
agents in $s_2$.
Third, the antagonist chooses a move for the adversarial agents in $s_2$,
and, fourth, the protagonist chooses a matching move for the adversarial
agents in $s_1$.
\end{longversion}
Alternating refinement holds if the protagonist has a winning strategy,
i.e., the protagonist is always able to match moves.
There are different ways of 'matching a move', and these determine
which alternating refinement relation is obtained.
The branching time \emph{alternating simulation} uses a local,
single transition-based notion of matching,
whereas the linear time \emph{alternating-trace-containment}
adopts a global, trace-based approach. 

A successful instantiation of alternating simulation is in interface theory,
where alternating transition systems are replaced by
\emph{Interface Automata} (IA) with inputs and outputs,
and where fixed agents are chosen:
the software system itself controls the outputs,
and the environment control the inputs.
This led to the definition of alternating simulation on
IA \cite{AlHe01}.
Unfortunately, alternating simulation is not black-box \emph{observational},
i.e., it is difficult to construct a realistic test and
observation scenario with which the differences between unrelated systems
can be observed in a black-box setting.
Since observable behaviour is often represented by trace-based
(linear time) relations, alternating-trace-containment for IA 
may be of interest, but this relation has not been translated to IA, yet.
A translation of alternating-trace-containment to labelled transition
systems with inputs and outputs was recently proposed \cite{BoSt18},
but only for deterministic models.


\paragraph{Relating Relations.}

The relations $\ioco$ and $\uioco$ on one hand,
and \emph{alternating-trace-containment} and \emph{alternating simulation}
on the other hand, were proposed in different communities
and for different purposes, yet, they show considerable overlap.
In particular, it has been shown that all four relations coincide
for deterministic models \cite{AaVa10,Aluetal98,BoSt18,VaBj10},
but for non-determinsitic models such a comparison has not been made, yet.
Only a conjecture in~\cite{BoSt18} claims that
alternating-trace-containment and $\ioco$ also coincide
for non-deterministic models.

If we manage to relate these independently defined relations
also for non-determinsitic systems, 
this would indicate that these relations indeed express
a generic and natural notion of conformance and refinement.
An integration of both paradigms would strengthen both
$\ioco$ theory and interface theory: it would add black-box observability
to alternating refinement, and it would provide concepts and algorithms
for refinement to $\ioco$/$\uioco$.


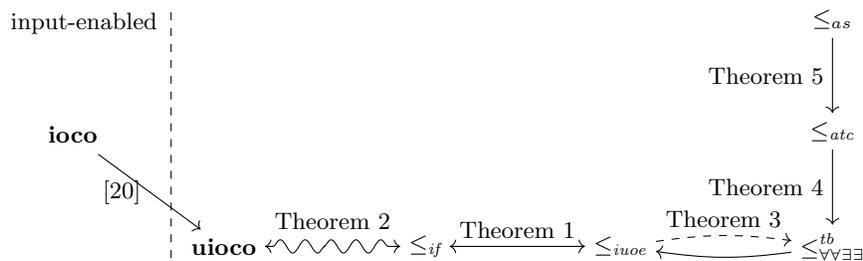
\begin{figure}[b!]
	\begin{tikzpicture}
	\node (ioco) {$\ioco$};
	\node (uioco) [below right=10mm and 10mm of ioco] {$\uioco$};
	\node (ir) [right=18mm of uioco] {$\ir$};
	\node (iuoe) [right=18mm of ir] {$\iuoe$};
	\node (aaee) [right=18mm of iuoe] {$\aaee$};
	\node (ati) [above=10mm of aaee] {$\ati$};
	\node (as) [above=10mm of ati] {$\as$};
	
	\node [above right=10mm and -14mm of ioco] {input-enabled};
	\path
	($(ioco)!0.65!(uioco) + (0,2.6)$) edge [dashed] ($(ioco)!0.65!(uioco) + (0,-0.7)$)
	;
	
	\path
	(ioco) edge [->] node [left] {\cite{Tre08}} (uioco)
	(uioco) edge [draw=none] node [above=1.5mm,inner sep=2pt,fill=white] {Theorem~\ref{the:uioco-ir}} (ir)
	($(uioco)!0.505!(ir)$) edge [->,decorate,decoration=snake,segment length=9.4] (uioco)
	($(uioco)!0.495!(ir)$) edge [->,decorate,decoration=snake,segment length=9.4] (ir)
	(ir) edge [<->] node [above] {Theorem~\ref{the:ir-iuoe}} (iuoe)
	(iuoe) edge [<-,bend right=8] (aaee)
	(iuoe) edge [->,dashed,bend left=8] node [above=0mm] {Theorem~\ref{the:iuoe-aaee}}(aaee)
	(aaee) edge [<-] node [left] {Theorem~\ref{the:ati-implies-aaee}} (ati)
	(ati) edge [<-] node [left] {Theorem~\ref{the:as-implies-ati}} (as)
	;
	
	\end{tikzpicture}
	\caption{
		Overview of the relations treated in this paper.
		An arrow from relation $A$ to relation $B$ denotes that $A$ is stronger than $B$.
		The dashed arrow only holds if the second of the related models is image finite.
		The wavy arrow only holds if quiescence is explicitly added to the models related by input-failure refinement.
		Relation $\ioco$ is only defined if the first argument is input-enabled.
	}
	\label{fig:overview}
\end{figure}	


\paragraph{Contributions and Overview.}

The main contribution of this paper is an integration of
the $\ioco$ theory of model-based testing \cite{Tre96},
the theory of interface automata \cite{AlHe01},
and the theory of alternating refinements \cite{Aluetal98}.
More specifically, we present the following results:

\begin{enumerate}

	\item
\emph{Input-failure refinement} $\ir$ and
\emph{input-universal output-existential refinement} $\iuoe$
are two equivalent preorders,
which are defined in Section~\ref{sec:input-failure-refinement}.
They are proved to coincide with $\uioco$,
after explicitly adding quiescence in the IAs, in Section~\ref{sec:uioco}.
The new preorders are in essence the same as the relation of substitutive refinement in~\cite{Chietal14}, but adapted to our context.
The new characterizations serve as a basis for the integration:
they increase intuition and understanding and they turn out to be helpful
for comparing with alternating refinements.

	\item 
The game-theoretic notion of alternating-trace containment
\cite{BoSt18,Aluetal98} is translated
to the setting of non-deterministic interface automata
in Section~\ref{sec:alternating-trace-inclusion}.
(Alternating simulation of \cite{Aluetal98} was already translated
in \cite{AlHe01}).
We show that the resulting alternating-trace containment preorder $\ati$
is weaker than alternating simulation preorder $\as$
for interface automata,
similar to the result of \cite{Aluetal98} for alternating transition
systems, in Section~\ref{sec:alternating-simulation}.

	\item 
We show that $\ati$ is not observational and not intuitive
as a conformance relation,
using a natural testing scenario for interface automata.
Motivated by this scenario, we define a slightly weaker game-theoretic
refinement relation $\aaee$.
We prove that $\aaee$ coincides with $\ir$ and $\iuoe$, and,
modulo proper treatment of quiescence, with $\uioco$,
for image-finite interface automata.
The tight link with $\uioco$ and $\ir$ implies that $\aaee$ is indeed
observational.
Moreover, these results disprove the conjecture that $\ioco$ and
alternating-trace containment coincide \cite{BoSt18}.

	\item 
We provide first steps towards a linear time -- branching time spectrum
for interface automata in Figure~\ref{fig:overview},
similar to the well-known linear time -- branching time spectrum
for labeled transition systems of Van Glabbeek \cite{vG01}.
Based on our classification, we motivate that also $\ioco$ is too strong to act as intuitive conformance relation.


\end{enumerate}
Recently, \cite{BoSt18} established a fundamental connection between model-based testing and 2-player concurrent games, in a setting of deterministic systems, where specifications are game arenas, test cases are game strategies, test case derivation is strategy synthesis, and conformance is alternating-trace containment. Our work show that the results of \cite{BoSt18} can be lifted to nondeterministic systems. This enables the application of a plethora of game synthesis techniques for test case generation.

\begin{shortversion}
Proofs and more extensive examples can be found in the technical report~\cite{JaVaTr19}.
\end{shortversion}
\begin{longversion}
  \vspace{10mm}
\end{longversion}

\section{Preliminaries}
\label{sec:preliminaries}

We start by introducing the basic definitions of state-based models with inputs and outputs, as a basis for the alternating relations as well as $\ioco$ and $\uioco$.
The former are defined on interface automata, whereas the latter are defined on labelled transition systems.
These paradigms differ mainly on the handling of internal transitions, which we omit in the scope of this paper.
Other differences are minor, so Definition~\ref{def:ia} reflects both  domains of models from both works.

\begin{definition}
  \label{def:ia}
  An Interface Automaton (IA) is a 5-tuple $(Q,I,O,T,q^0)$, where
  \begin{itemize}
    \item $Q$ is a set of states,
    \item $I$ and $O$ are disjoint sets of input and output labels,
    \item $T \subseteq Q \times (I \cup O) \times Q$ is a transition relation, and
    \item $q^0 \in Q$ is the initial state.
  \end{itemize}
  
  The domain of IA is denoted $\lts$.
  For $s \in \lts$, we write $Q_s$, $I_s$, $O_s$, $T_s$ and $q_s^0$ to refer to its respective elements, and $L_s = I_s \cup O_s$ for the full set of labels.
  For $s_1, s_2,\ldots, s_A, s_B, \dots$ a family of IAs, we write $Q_j$, $I_j$, $O_j$, $T_j$ and $q_j^0$ to refer to the respective  elements, and $L_j = I_j \cup O_j$, for $j = 1, 2,\ldots, A, B, \dots$.
\end{definition}

In examples, we represent IA as state diagrams as usual.
For the remainder of this paper, we assume that IA have the same input alphabet $I$ and output alphabet $O$, with $L = I \cup O$, unless explicitly stated otherwise.
Symbols $a$ and $b$ represent inputs, and $x$, $y$ and $z$ represent outputs.

\begin{definition}
  \label{def:lts-definitions}
  Let $s \in \lts$, $Q \subseteq Q_s$, $q,q' \in Q_s$, $\ell \in L$ and $\sigma \in L^*$, where $^*$ denotes the Kleene star, and $\epsilon$ denote the empty sequence.
  We define
  \begin{align*}
    q \xrightarrow{\epsilon}_s q'\:\: &\ifdefsmall q = q' &
    q \xrightarrow{\sigma\ell}_s q'\:\: &\ifdefsmall \exists r \in Q_s: q \xrightarrow{\sigma}_s r \wedge (r,\ell,q') \in T_s\\
    q \xrightarrow{\sigma}_s \:\: &\ifdefsmall \exists r \in Q_s: q \xrightarrow{\sigma}_s r &
    \trans(s,q) &\isdef \{(r,\ell,r') \in T_s \mid q = r\}\\
    \traces_s(q) &\isdef \{\sigma \in L^* \mid q \xrightarrow{\sigma}_s\} &
    \after[s]{Q}{\sigma} &\isdef \{r \in Q_s \mid \exists r' \in Q: r' \xrightarrow{\sigma} r\} \\
    \traces(s) &\isdef \traces_s(q_s^0) &
    \after{s}{\sigma} &\isdef \after[s]{\{q_s^0\}}{\sigma}\\
    \out_s(q) &\isdef \{x \in O \mid q \xrightarrow{x}_s\} &
    \out_s(Q) &\isdef \{x \in O \mid \exists q \in Q: x \in \out_s(q)\}\\
    \inp_s(q) &\isdef \{a \in I \mid q \xrightarrow{a}_s\} &
    \inp_s(Q) &\isdef \{a \in I \mid \forall q \in Q: a \in \inp_s(q)\}\\
    \multispan{4}{\quad$\text{$s$ is deterministic} \:\:\iff\:\: \forall \sigma \in \traces(s): |\after{s}{\sigma}| = 1$\hfil}&&&
  \end{align*}
\end{definition}

We omit the subscript for interface automaton $s$ when clear from the context.

\begin{definition}
  \label{def:paths}
  For $s \in \lts$, a \emph{path} through $s$ is a (finite or infinite) sequence $\pi = q^0 \ell^1 q^1 \ell^2 q^2 \cdots$ of alternating states from $Q_s$ and labels from $L$
  starting with state $q^0 = q_s^0$ and, if the sequence is finite, also ending in a state,
   such that each triplet $(q^j, \ell^{j+1}, q^{j+1})$ is contained in $T_s$.
  The domain of finite paths through $s$ is denoted $\Path(s)$.
  The \emph{trace} of path $\pi$ is the subsequence of labels that occur in it: $\trace(\pi) \isdef \ell^1 \ell^2 \cdots$.
  Note that each $\pi \in \Path(s)$ has $\trace(\pi) \in \traces(s)$.
  We write $\last(\pi)$ to denote the last state occurring in a finite path $\pi$.
\end{definition}

\section{Two Preorders on Interface Automata}
\label{sec:input-failure-refinement}

We now present two equivalent relations, which serve as a stepping stone to bridge the gap between $\ioco$ theory and alternating refinements.
The first relation has a clear observational interpretation, whereas the second is more elegant and convenient in reasoning, which turns out useful in proofs and examples.
They essentialy coincide with the relation of substitutive refinement in~\cite{Chietal14}, adapted to our context.

\subsection{Input-Failure Refinement}

The first relation is based on covariance and contravariance, strongly inspired by interface theory~\cite{AlHe01}.
Outputs are treated covariantly, as in normal trace containment: if $s_1$ refines $s_2$, then outputs produced by $s_1$ are also produced by $s_2$.
Inputs are treated contravariantly, instead: inputs refused by $s_1$ are also refused by $s_2$.
If $s_2$ refuses an input, then $s_1$ may choose to be more liberal than $s_2$, accepting that input instead.
If it does so, the behaviour after that input is unspecified.

We first make the notion of refusing an input explicit.

\begin{definition}
  \label{def:refusal}
  For any input symbol $a$, we define the \emph{input-failure of $a$} as $\overline a$.
  Likewise, for any set of inputs $A$, we define $\overline A \isdef \{\overline a \mid a \in A\}$.
  The domain of \emph{input-failure traces} is defined as 
  $\iftracesdom_{I,O} \isdef[0mm] L^* \cup L^* \cdot \overline I$.
  Set $S \subseteq \iftracesdom_{I,O}$ of input-failure traces is \emph{input-failure closed} if, for all $\sigma \in L^*$, $a \in I$ and $\rho \in \iftracesdom_{I,O}$,
    $ \sigma\overline a \in S \implies \sigma a \rho \in S$.
    The \emph{input-failure closure} of $S$ is the smallest input-failure closed superset of $S$, that is, $\rcl(S) \isdef S \cup \{\sigma a \rho \mid \sigma\overline a \in S, \rho \in \iftracesdom_{I,O}\}$.
\end{definition}

We associate with every IA a set of input-failure traces, to define \emph{input-failure refinement} and \emph{input-failure equivalence}, denoted $\ir$ and  $\ireq$, respectively.

\begin{definition}
  \label{def:ir-inclusion}
  Let $s_1,s_2 \in \lts$. Then
  \begin{align*}
    \irtraces(s_1) &\isdef \traces(s_1) \cup \{\sigma \overline a \mid \sigma \in L^*, 
    a \in I, a \not \in \inp(\after{s_1}{\sigma})
    \}\\
    s_1 \ir s_2 &\ifdef \irtraces(s_1) \subseteq \rcl(\irtraces(s_2))\\
    s_1 \ireq s_2 &\ifdef s_1 \ir s_2 \wedge s_2 \ir s_1
  \end{align*}
\end{definition}

Remark that the definition of $\inp(Q)$ is universally quantified over states in $Q$.
This means that $\sigma \overline{a}$ is an input-failure trace of $s$ if $\sigma \in \traces(s)$ and some state in $\after{s}{\sigma}$ refuses input $a$.
For a closed system, all actions are outputs.
In that case, input-failure refinement coincides with ordinary trace containment.

The $\irtraces$ provide an observational, trace-based semantics for IA.
Intuitively, to observe an input-failure trace of system $s_1$, we let $s_1$ interact with its environment.
The system produces outputs and consumes inputs from the environment, until we decide to stop, or until the system refuses an input.
If the resulting input-failure trace is not in the closure of the $\irtraces$ of specification $s_2$, then it proves $s_1 \not\ir s_2$.
If no such trace can be found, then $s_1 \ir s_2$ holds.

\begin{example}
 \label{exa:ir}
 Figure~\ref{fig:ir} shows four IA with $I = \{a\}$ and $O = \{x\}$.
 Clearly, $\overline{a} \not\in \irtraces(s_A)$ holds.
 Furthermore, $ax\overline{a} \not\in \irtraces(s_A)$ since following trace $ax$ in $s_A$ leads to state $q_A^1$, and input $a$ is not refused in $q_A^1$.
 In contrast, $a\overline{a} \in \irtraces(s_A)$ holds, since $(\after{s_A}{a}) = \{q_A^1,q_A ^2\}$, and input $a$ is not enabled in $q_A^2$. 
 
 Now, let us establish whether $s_B$, $s_C$ and $s_D$ are input-failure refinements of $s_A$.
 We find $s_B \not\ir s_A$, shown by trace $axax \in \irtraces(s_B)$ which is not allowed by $s_A$.
 Put formally, $axax \not\in\rcl(\irtraces(s_A))$ holds, because $axax$, $ax\overline{a}$ and $\overline{a}$ are not in $\irtraces(s_A)$.
 Similarly $s_C \not\ir s_A$ is shown by trace $ax\overline{a} \in \irtraces(s_C)$.
 Refusing $\overline{a}$ after $ax$ is not allowed by $s_A$, or formally, $ax\overline{a} \not\in\rcl(\irtraces(s_A))$.
 Finally, $s_D \ir s_A$ does hold, as $\irtraces(s_D) = \{\epsilon,a,aa,aax\} \cup aax^*\overline{a}$, and all of these traces are in $\rcl(\irtraces(s_A))$: traces $\epsilon$ and $a$ are in $\rcl(\irtraces(s_A))$ because they are in $\irtraces(s_A)$, and all other traces are in $\rcl(\irtraces(s_A))$ because $a\overline{a} \in \irtraces(s_A)$.
 Intuitively, $a\overline{a} \in \irtraces(s_A)$ implies that the behaviour after trace $aa$ is underspecified, so $s_D$ is free to choose any behaviour after this trace.
  
 \begin{figure}
  \tikzlts
  \begin{center}
   \begin{tikzpicture}[node distance=17mm]
     \node (center) {};
     \node [named,initial] (0) {$q_A^0$};
     \node (init-text) [above left=-1.7mm and 0 of 0] {$s_A$};
     \node [named] (1) [above right of=0] {$q_A^1$};
     \node [named] (2) [below right of=0] {$q_A^2$};
     \path
     (0) edge [loop below] node {$x$} (0)
     (0) edge node [above left=0mm and -1mm] {$a$} (1)
     (0) edge node [below left=0mm and -1mm] {$a$} (2)
     (1) edge [loop right] node {$x$} (1)
     (1) edge node [right] {$a$} (2)
     ;
     
     \node [above left=3mm and 8mm of center] {\rotatebox{-20}{$\notrel\ir$}};
     
     \node [named,initial] [above left=8mm and 22mm of center] (0) {$q_B^0$};
     \node (init-text) [above left=-1.7mm and 0 of 0] {$s_B$};
     \path
       (0) edge [loop below] node {$x$} (0)
       (0) edge [in=15,out=45,loop] node [right] {$a$} (0)
     ;
     
     \node [below left=0mm and 8mm of center] {\rotatebox{20}{$\notrel\ir$}};
     
     \node [named,initial,initial text=] [below left=8mm and 35mm of center] (0) {$q_C^0$};
     \node (init-text) [above left=-1.7mm and 0 of 0] {$s_C$};
     \node [named] (1) [right of=0] {$q_C^1$};
     \path
     (0) edge [loop above] node {$x$} (0)
     (1) edge [loop above] node {$x$} (1)
     (0) edge node [above] {$a$} (1)
     ;
     
     \node [above right=-2mm and 20mm of center] {$\irrev$};
     
     \node [named,initial,initial text=] [above right=8mm and 35mm of center] (0) {$q_D^0$};
     \node (init-text) [above left=-1.7mm and 0 of 0] {$s_D$};
     \node [named] (1) [below=6mm of 0] {$q_D^1$};
     \node [named] (2) [below=6mm of 1] {$q_D^2$};
     \path
     (0) edge node [right] {$a$} (1)
     (1) edge node [right] {$a$} (2)
     (2) edge [loop right] node {$x$} (2)
     ;
   \end{tikzpicture}
  \end{center}
  \caption{Specification IA $s_A$ is not input-failure refined by $s_B$ and $s_C$, but it is by $s_D$.}
  \label{fig:ir}
 \end{figure}
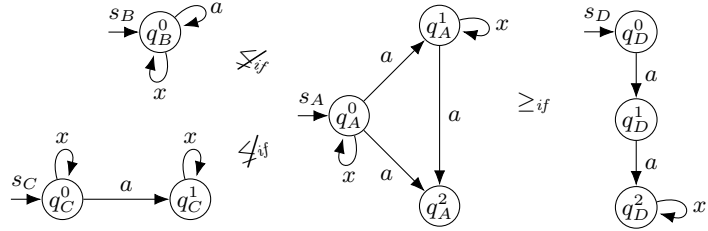
\end{example}


The closure of the input-failure traces serves as a canonical representation of the behaviour of an IA, as stated in Proposition~\ref{pro:rcl-irtraces-canonical}.
That is, if and only if the closures of two models are the same, then they are input-failure equivalent.

\begin{proposition}
  \label{pro:rcl-irtraces-canonical}
  Let $s_1, s_2 \in \lts$.
  Then 
  \begin{align*}
    s_1 \ir s_2 \iff& \rcl(\irtraces(s_1)) \subseteq \rcl(\irtraces(s_2))\\
    s_1 \ireq s_2 \iff& \rcl(\irtraces(s_1)) = \rcl(\irtraces(s_2))
  \end{align*}
\end{proposition}

\begin{proof}
  The latter statement follows from the former.
  We now prove the former.
  
  $(\implies)$
  Let $s_1 \ir s_2$.
  That is, $\irtraces(s_1) \subseteq \rcl(\irtraces(s_2))$, or put differently, $\rcl(\irtraces(s_2))$ is an input-failure closed superset of $\irtraces(s_1)$.
  Then $\rcl(\irtraces(s_2))$ must be larger than the smallest input-failure closed superset of $\irtraces(s_1)$, which is $\rcl(\irtraces(s_1))$.
  
  $(\impliedby)$
  This follows trivially from the fact that $\irtraces(s_1) \subseteq \rcl(\irtraces(s_1))$, and from transitivity of $\subseteq$.
  \qed
\end{proof}

Proposition~\ref{pro:rcl-irtraces-canonical} implies that relation $\ir$ is reflexive ($s \ir s$), so any software component may safely be replaced by an input-failure equivalent one.
Relation $\ir$ is also transitive ($s_1 \ir s_2 \wedge s_2 \ir s_3 \implies s_1 \ir s_3$), making it suitable for stepwise refinement.
Formally, it is thus a preorder.

\subsection{Input Universal / Output Existential Traces}
\label{sec:input-universal-output-existential}

The definition of input-failure refinement clearly reflects its observational nature.
Yet, reasoning about this relation can be simplified by using an alternative characterization.
This characterization is not expressed in terms of explicit input refusals, but it is based upon the existential and universal definitions of the respective operators $\out$ and $\inp$, from Definition~\ref{def:lts-definitions}.\begin{longversion}
\footnote{
not to be confused with the existential and universal quantifications of inputs and outputs in the interfaces of~\cite{Chaetal02}, which have a different meaning.}\end{longversion}

\begin{longversion}
Some auxiliary definitions and lemmas are introduced, before providing the characterization in Theorem~\ref{the:ir-iuoe}.
\end{longversion}

\begin{definition}
  \label{def:output-existential}
  \label{def:input-universal}
  Let $s \in \lts$.
  A word $\ell^1 \dots \ell^n \in L^*$ is \emph{$s$-output-existential} if 
  \[\forall j \in \{1 \dots n\}: \ell^j \in O \hspace{-1mm}\implies\hspace{-1mm} \ell^j \in \out(\after{s}{\ell^1 \dots \ell^{j-1}})\]
  and it is \emph{$s$-input-universal} if
  \[\forall j \in \{1 \dots n\}: \ell^j \in I \implies \ell^j \in \inp(\after{s}{\ell^1 \dots \ell^{j-1}})\]
  $\oetraces(s)$ denotes the set of  $s$-output-existential words in $L^*$, and
  $\iutraces(s)$ the set of $s$-input-universal words in $L^*$.
\end{definition}

\begin{longversion}
Note that the sets $\oetraces(s)$ and $\iutraces(s)$ are both prefix closed.
\end{longversion}

\begin{longversion}
\begin{lemma}
  \label{lem:oe-traces-form}
  Let $s \in \lts$.
  Then $\oetraces(s) = \traces(s) \cdot I^*$.
\end{lemma}

\begin{proof}
  For $\sigma \in L^*$, we prove $\sigma \in \oetraces(s) \iff \sigma \in \traces(s) \cdot I^*$.
  
  $(\implies)$
  Assume $\sigma \in \oetraces(s)$.
  We prove $\sigma \in \traces(s) \cdot I^*$ by a case distinction:
  \begin{itemize}
    \item
      If $\sigma \in I^*$, then $\sigma \in \traces(s) \cdot I^*$ trivially holds, since $\epsilon \in \traces(s)$.
    \item
      If $\sigma \not\in I^*$, then $\sigma$ containts at least one output symbol. Let $x \in O$ be the last output symbol that occurs in $\sigma$. Then $\sigma = \rho \, x \, \tau$ for some $\tau \in I^*$.
      Since $\sigma \in \oetraces(s)$, also $\rho \, x \in \oetraces(s)$. By Definition~\ref{def:output-existential} of $\oetraces$ this implies $q_s^0 \xrightarrow{\rho} q \xrightarrow{x}$, for some $q$.
      Then $\rho \, x \in \traces(s)$, so $\sigma = \rho \, x \, \tau \in \traces(s) \cdot I^*$.
  \end{itemize}

  $(\impliedby)$
  Assume $\sigma \in \traces(s) \cdot I^*$. Then $\sigma = \rho\,\tau$, for some $\rho \in \traces(s)$ and $\tau \in I^*$.
  We first prove $\rho \in \oetraces(s)$ by induction on the length of $\rho$.
  For the base case $\rho = \epsilon$, $\rho \in \oetraces(s)$ trivially holds.
  For the inductive step, let $\rho = \rho'\, \ell$ and assume as induction hypothesis that $\rho' \in \oetraces(s)$.
  Then we distinguish two cases:
  \begin{itemize}
    \item
      If $\ell \in I$, then $\rho' \in \oetraces(s)$ implies $\rho \in \oetraces(s)$.
    \item
      If $\ell \in O$, then $\rho \in \traces(s)$ implies $q_s^0 \xrightarrow{\rho'} q \xrightarrow{\ell}$, for some $q$, so together with the induction hypothesis this implies $\rho \in \oetraces(s)$.
  \end{itemize}
  Thus, $\rho \in \oetraces(s)$ holds.
  Since $\sigma = \rho \, \tau$ with $\tau \in I^*$, this implies $\sigma \in \oetraces(s)$.
  \qed
\end{proof}

\begin{lemma}
  \label{lem:iu-trace-not-refused}
  Let $s \in \lts$, and $\sigma \in L^*$.
  Then 
  \begin{align*}
    \sigma \in \iutraces(s) \iff& \text{no decomposition $\sigma = \rho \, a \, \tau$}\\&\text{with $a \in I$ has $\rho \overline{a} \in \irtraces(s)$}
  \end{align*}
\end{lemma}

\begin{proof}
  \begin{align*}
        & \sigma \in \iutraces(s)\\
    \iff& \text{all decompositions $\sigma = \rho \, a \, \tau$ with $a \in I$ have $a \in \inp(\after{s}{\rho})$} \ttag{Definition~\ref{def:input-universal}}\\
    \iff& \text{no decomposition $\sigma = \rho \, a \, \tau$ with $a \in I$ has $\exists q \in (\after{s}{\rho}): q \not\xrightarrow{a}$} \ttag{Definition~\ref{def:lts-definitions} of $\inp$}\\
    \iff& \text{no decomposition $\sigma = \rho \, a \, \tau$ with $a \in I$ has $\rho \overline{a} \in \irtraces(s)$} \\&\qquad \ttag{Definition~\ref{def:ir-inclusion} of $\irtraces$}
  \end{align*}
  \qed
\end{proof}

\begin{lemma}
  \label{lem:intersection-irtraces-iutraces}
  Let $s \in \lts$.
  Then
  $\iutraces(s) \cap \rcl(\irtraces(s)) \subseteq \traces(s)$.
\end{lemma}

\begin{proof}
  For $\sigma \in \iutraces(s) \cap \rcl(\irtraces(s))$ we prove $\sigma \in \traces(s)$:
  \begin{align*}
            & \sigma \in \iutraces(s) \cap \rcl(\irtraces(s))\\
    \implies& \text{no decomposition $\sigma = \rho \, a \, \tau$ with $a \in I$ has $\rho \overline{a} \in \irtraces(s)$} \\&\qquad \wedge \sigma \in \rcl(\irtraces(s)) \ttag{Lemma~\ref{lem:iu-trace-not-refused}}\\
    \implies& \sigma \in \traces(s) \ttag{Definition~\ref{def:ir-inclusion} of $\rcl$}
  \end{align*}
  \qed
\end{proof}

\begin{lemma}
  \label{lem:intersection-oe-iu-in-traces}
  Let $s \in \lts$.
  Then $\iutraces(s) \cap \oetraces(s) = \iutraces(s) \cap \traces(s)$.
\end{lemma}

\begin{proof}
  $(\subseteq)$
  Assume $\sigma \in \iutraces(s) \cap \oetraces(s)$.
  We prove $\sigma \in \traces(s)$ by induction on the length of $\sigma$.
  For the base case $\sigma = \epsilon$, $\sigma \in \traces(s)$ trivially holds.
  For the inductive step, assume $\sigma = \rho \, \ell$ with $\rho \in \traces(s)$.
  We distinguish two cases:
  \begin{itemize}
    \item
      If $\ell \in I$, then
      \begin{align*}
                & \rho \in \traces(s)\\
        \implies& \text{there is some $q \in (\after{s}{\rho})$}\\
        \implies& q_s^0 \xrightarrow{\rho} q \xrightarrow{\ell} \ttag{$\ell \in I$ and $\sigma \in \iutraces(s)$}\\
        \implies& \sigma \in \traces(s)
      \end{align*}
    \item
      If $\ell \in O$, then $\sigma = \rho\, \ell \in \oetraces(s)$ implies that there is a $q$ such that $q_s^0 \xrightarrow{\rho} q \xrightarrow{\ell}$ by Definition~\ref{def:output-existential} of $\oetraces(s)$, so $\sigma \in \traces(s)$ holds.
  \end{itemize}
  If $\sigma \in \traces(s)$, then clearly also $\sigma \in \iutraces(s) \cap \traces(s)$.
  
  $(\supseteq)$
  This follows directly from Lemma~\ref{lem:oe-traces-form}.
  \qed
\end{proof}
\end{longversion}

\begin{definition}
  \label{def:iuoe}
  Let $s_1, s_2 \in \lts$.
  Then 
  \[s_1 \iuoe s_2 \ifdef \oetraces(s_1) \cap \iutraces(s_2) \subseteq \iutraces(s_1) \cap \oetraces(s_2)\]
\end{definition}

\begin{theorem}
  \label{the:ir-iuoe}
  $s_1 \ir s_2 \iff s_1 \iuoe s_2$
\end{theorem}

\begin{proof}
  $(\implies)$
  Assume $s_1 \ir s_2$ (1) and $\sigma \in \oetraces(s_1) \cap \iutraces(s_2)$ (2).
  We prove $\sigma \in \iutraces(s_1) \cap \oetraces(s_2)$.
  \begin{align*}
            & \sigma \in \iutraces(s_2) \ttag{assumption (2)}\\
    \implies& \text{No decomposition $\sigma = \rho \, a \,\tau$ with $a \in I$ has $\rho \overline{a} \in \irtraces(s_2)$} \ttag{Lemma~\ref{lem:iu-trace-not-refused}}\\
    \implies& \text{All decompositions $\sigma = \rho\,  a \, \tau$ with $a \in I$ have $\rho \overline{a} \not\in \irtraces(s_2)$}\\
    \implies& \text{All decompositions $\sigma = \rho \, a \, \tau$ with $a \in I$ have $\rho \overline{a} \not\in \rcl(\irtraces(s_2))$} \ttag{Definition~\ref{def:ir-inclusion} of $\rcl$}\\
    \implies& \text{All decompositions $\sigma = \rho \, a \, \tau$ with $a \in I$ have $\rho \overline{a} \not\in \irtraces(s_1)$} \ttag{assumption (1) and Definition~\ref{def:ir-inclusion} of ${\ir}$}\\
    \implies& \text{No decomposition $\sigma = \rho \, a \, \tau$ with $a \in I$ has $\rho \overline{a} \in \irtraces(s_1)$}\\
    \implies& \sigma \in \iutraces(s_1) \quad(3)\quad \text{ and} \ttag{Lemma~\ref{lem:iu-trace-not-refused}}
    \\&\qquad\text{no decomposition $\sigma = \rho \, a \, \tau$ with $a \in I$ has $\rho \in \traces(s_1)$ and }\\&\qquad \rho a \not\in \traces(s_1) \ttag{Definition~\ref{def:ir-inclusion} of $\irtraces$}\\
    \implies& \text{No decomposition $\sigma = \rho \, a \, \tau$ has $\rho \in \traces(s_1)$, $a \tau \in I^*$, $\rho a \not\in \traces(s_1)$}\\
    \implies& \sigma \not\in (\traces(s_1) \cdot I^*) \setminus \traces(s_1)\\
    \implies& \sigma \not\in \oetraces(s_1) \setminus \traces(s_1) \ttag{Lemma~\ref{lem:oe-traces-form}}\\
    \implies& \sigma \in \traces(s_1) \ttag{$\sigma \in \oetraces(s_1)$ by assumption (2)}\\
    \implies& \sigma \in \traces(s_1) \cap \iutraces(s_2) \ttag{$\sigma \in \iutraces(s_2)$ by assumption (2)}\\
    \implies& \sigma \in \irtraces(s_1) \cap \iutraces(s_2) \ttag{Definition~\ref{def:ir-inclusion} of $\irtraces$}\\
    \implies& \sigma \in \rcl(\irtraces(s_2)) \cap \iutraces(s_2) \ttag{assumption (1) and Definition~\ref{def:ir-inclusion}}\\
    \implies& \sigma \in \traces(s_2) \ttag{Lemma~\ref{lem:intersection-irtraces-iutraces}}\\
    \implies& \sigma \in \oetraces(s_2) \ttag{$\traces(s_2) \subseteq \oetraces(s_2)$ by Lemma~\ref{lem:oe-traces-form}} \quad (4)
  \end{align*}
  From (3) and (4) we conclude  that $\sigma \in \iutraces(s_1) \cap \oetraces(s_2)$, as required, which proves $s_1 \iuoe s_2$.
  
  $(\impliedby)$
  Assume $s_1 \iuoe s_2$ (1) and $\sigma \in \irtraces(s_1)$ (2).
  We prove $s_1 \ir s_2$ by showing $\sigma \in \rcl(\irtraces(s_2))$.
  We distinguish three cases:
  \begin{itemize}
    \item
      If $\sigma \in L^* \setminus \iutraces(s_2)$, then
      \begin{align*}
		& \sigma \not\in \iutraces(s_2)\\
	\implies& \text{there exists a decomposition $\sigma = \rho\, a\, \tau$ with $a \in I$}\\&\qquad\text{and $\rho \overline{a} \in \irtraces(s_2)$} \ttag{Lemma~\ref{lem:iu-trace-not-refused}}\\
	\implies& \text{there exists a decomposition $\sigma = \rho\, a \,\tau$ with $a \in I$}\\&\qquad \text{and $\rho \, a \,\tau \in \rcl(\irtraces(s_2))$} \ttag{Definition~\ref{def:ir-inclusion} of $\rcl$}\\
	\implies& \sigma \in \rcl(\irtraces(s_2))
      \end{align*}
    \item 
      If $\sigma \in \iutraces(s_2)$, then
      \begin{align*}
                & \sigma \in \iutraces(s_2) \cap L^* \ttag{$\iutraces(s_2) \subseteq L^*$}\\
	\implies& \sigma \in \iutraces(s_2) \cap \traces(s_1) \ttag{(2) and Definition~\ref{def:ir-inclusion} of $\irtraces$}\\
	\implies& \sigma \in \iutraces(s_2) \cap \oetraces(s_1) \ttag{$\traces(s_1) \subseteq \oetraces(s_1)$ by Lemma~\ref{lem:oe-traces-form}}\\
	\implies& \sigma \in \iutraces(s_2) \cap \oetraces(s_2) \ttag{assumption (1)}\\
	\implies& \sigma \in \traces(s_2) \ttag{Lemma~\ref{lem:intersection-oe-iu-in-traces}}\\
	\implies& \sigma \in \irtraces(s_2) \ttag{Definition~\ref{def:ir-inclusion} of $\irtraces$}\\
	\implies& \sigma \in \rcl(\irtraces(s_2)) \ttag{Definition~\ref{def:ir-inclusion} of $\rcl$}
      \end{align*}
    \item
      If $\sigma \not\in L^*$, then 
      \begin{align*}
                & \sigma \not\in L^*\\
        \implies& \sigma = \rho\overline{a} \text{ for some $a \in I$ and $q \in Q_s$ with $q_1^0 \xrightarrow{\rho} q \not\xrightarrow{a}$} \ttag{assumption (2) and Definition~\ref{def:ir-inclusion} of $\irtraces$}\\
        \implies& \sigma = \rho\overline{a} \wedge \rho \in \traces(s_1) \wedge \rho a \not\in \iutraces(s_1) \ttag{assumption (2) and Lemma~\ref{lem:iu-trace-not-refused}}\\
        \implies& \sigma = \rho\overline{a} \wedge \rho \in \traces(s_1) \wedge \rho a \not\in (\iutraces(s_1) \cap \oetraces(s_2))\\
        \implies& \sigma = \rho\overline{a} \wedge \rho \in \traces(s_1) \wedge \rho a \not\in (\oetraces(s_1) \cap \iutraces(s_2)) \ttag{assumption (1)}\\
        \implies& \sigma = \rho\overline{a} \wedge \rho \in \traces(s_1)\wedge \rho a \in \oetraces(s_1) \wedge \rho a \not\in (\oetraces(s_1) \cap \iutraces(s_2))\ttag{Lemma~\ref{lem:oe-traces-form}}\\
        \implies& \sigma = \rho\overline{a} \wedge \rho \in \traces(s_1)\wedge \rho a \not\in \iutraces(s_2)\\
        \implies& \sigma \in \irtraces(s_2) \ttag{Lemma~\ref{lem:iu-trace-not-refused}}\\
        \implies& \sigma \in \rcl(\irtraces(s_2)) \ttag{Definition~\ref{def:ir-inclusion} of $\rcl$}
      \end{align*}
  \end{itemize}
  \qed
\end{proof}

\begin{example}
  \label{exa:iu-oe}
  We revisit Example~\ref{exa:ir}, and we should find the same IA to be related by $\iuoe$ as by $\ir$, by Theorem~\ref{the:ir-iuoe}.
  We find $axax \in \oetraces(s_B) \cap \iutraces(s_A)$ and $axax \not\in\oetraces(s_A)$, which confirms $s_B \not\iuoe s_A$.
  We also find $axa \in \oetraces(s_C) \cap \iutraces(s_A)$ and $axa \not\in\iutraces(s_C)$, confirming $s_C \not\iuoe s_A$.
  Finally, we find $\oetraces(s_D) \cap \iutraces(s_A) = \{\epsilon, a\}$, and these traces are both in $\iutraces(s_D)$ and in $\oetraces(s_A)$, which confirms $s_D \iuoe s_A$.
\end{example}

\section{Characterizing \emph{uioco}}
\label{sec:uioco}


An often used implementation relation for MBT on interface automata
(or labelled transition systems) is $\ioco$ \cite{Tre96,Tre08}.
\begin{shortversion}
For $\ioco$ it is assumed that implementations can be modelled as
\emph{input-enabled} interface automata, denoted by $\iots$.
\end{shortversion}
\begin{longversion}
For $\ioco$ it is assumed that implementations can be modelled as
\emph{input-enabled} interface automata, denoted by $\iots$ (testability hypothesis).
\end{longversion}
Moreover, quiescence is assumed to be observable.
Formally, quiescence is expressed by adding a fresh output
label $\delta \not \in L_s$ in all states where no outputs are possible.
(Since this changes the output alphabet, we will not assume the globally defined alphabets $L$, $I$ and $O$ for the remainder of this section.)

\begin{definition}
  \label{def:ioco}
Let $i \in \iots$ and $s \in \lts$
with $I_i = I_s$, $O_i = O_s$, and $\delta \not \in L_s$.
\begin{enumerate}
\item
  $\iots \isdef \{ s \in \lts \mid
  \forall q \in Q_s, \forall a \in I_s: q \xrightarrow{a} \}$
\item
  $\begin{array}[t]{lll}
    \Delta(s) & \isdef &
    (Q_s, I_s, O_s \cup \{\delta\}, T_\delta, q_s^0) \in \lts
    \text{ , with} \\
    T_\delta & \isdef &
    T_s \cup \{ (q,\delta,q) \mid q \in Q_s, \out(q) = \emptyset \} \\
  \end{array}$
\item
   $ i \ioco s \ifdef
    \forall \sigma \in \traces(\Delta(s)):
    \out(\after{\Delta(i)}{\sigma}) \subseteq
    \out(\after{\Delta(s)}{\sigma})$
\end{enumerate}
\end{definition}

A variation of $\ioco$ is $\uioco$ \cite{BiReTr04}.
Whereas $\ioco$ quantifies over all possible traces (with quiescence)
in $\traces(\Delta(s))$, including those where some input in the trace
may be underspecified, $\uioco$ only considers traces where all inputs
are never underspecified.
We take the generalized definition from \cite{VoTr13},
which also applies to non-input-enabled implementations.
This definition coincides with the original one \cite{BiReTr04}
if restricted to input-enabled implementations.

\begin{definition}
  \label{def:uioco}
Let $i,s \in \lts$ with $I_i = I_s$ and $O_i = O_s$, and let $\preccurlyeq$ denote the prefix relation on traces.
\begin{enumerate}
\item
  $\begin{array}[t]{llll}
    \utraces(s) & \isdef &
    \{ \sigma \in \traces(s) \mid &
    \forall \rho \in L_s^*, a \in I_s :
    \rho\,a \preccurlyeq \sigma \;\implies\; \\
    & & & a \in \inp(\after{s}{\rho}) \} \\
  \end{array}$
\item
  $\begin{array}[t]{lll}
  i \uioco s & \ifdef &
  \forall \sigma \in \utraces(\Delta(s)) : \\
    & & \begin{array}[t]{llll}
        & \out(\after{\Delta(i)}{\sigma}) & \subseteq &
          \out(\after{\Delta(s)}{\sigma}) \\
        \;\wedge\; & \inp(\after{\Delta(i)}{\sigma}) & \supseteq &
                     \inp(\after{\Delta(s)}{\sigma}) \\
        \end{array}
  \end{array}$
\end{enumerate}
\end{definition}

\begin{proposition}
  \label{pro:uioco}
\begin{shortversion}
  All $i \in \iots$ and $s \in \lts$ have $i \ioco s \;\implies\; i \uioco s$.
  There exist $i \in \iots$ and $s \in \lts$ with $i \uioco s \;{\notrel[6.5pt]\implies}\; i \ioco s$.
\end{shortversion}
\begin{longversion}
Let $i \in \iots$ and $s \in \lts$.
\begin{enumerate}
  \item
    $i \uioco s \iff  
    \forall \sigma \in \utraces(\Delta(s)) :
    \out(\after{\Delta(i)}{\sigma}) \subseteq
    \out(\after{\Delta(s)}{\sigma})$
\item
    $i \ioco s \;\implies\; i \uioco s$;~ $\ioco \;\neq\; \uioco$
\end{enumerate}
\end{longversion}
\end{proposition}

\begin{proof}
  This follows directly from the definitions.
\end{proof}

The next step is to relate $\ioco$ and $\uioco$ to the relations defined
in the previous sections.
The main result of this section is that 
$\uioco$ is the same as input-failure refinement
(Theorem~\ref{the:uioco-ir}),
and thus also as input-universal-output-existential refinement
(Theorem~\ref{the:ir-iuoe}),
if in the latter quiescence is explicitly added.
The consequence is that $\ioco$ and input-failure refinement
do not coincide, following Proposition~\ref{pro:uioco}.
The difference between the two relations is the treatment
of specification traces which are not input-universal,
as shown in Example~\ref{exa:ioco-not-ir}.

%

\begin{lemma}
  \label{lem:utraces-iutraces}
  $\utraces(s) = \iutraces(s) \cap \traces(s)$.
\end{lemma}

\begin{proof}
  Immediate from Definition~\ref{def:uioco} of $\utraces$ and Definition~\ref{def:input-universal} of $\iutraces$.
  \qed
\end{proof}

\begin{theorem}
  \label{the:uioco-ir}
    $i \uioco s \;\iff\; \Delta(i) \;\ir\; \Delta(s)$
\end{theorem}

\begin{proof}
  By Theorem~\ref{the:ir-iuoe}, it suffices to prove $i \uioco s \iff \Delta(i) \iuoe \Delta(s)$.
  
  $(\implies)$
  Assume $i \uioco s$ (1) and $\sigma \in \oetraces(\Delta(i)) \cap \iutraces(\Delta(s))$ (2).
  We must prove $\sigma \in \iutraces(\Delta(i)) \cap \oetraces(\Delta(s))$.
  The proof is by induction on the length of $\sigma$.
  
  For the base case $\sigma = \epsilon$, $\sigma \in \iutraces(\Delta(i))$ and $\sigma \in \oetraces(\Delta(s))$ hold trivially.
  For the inductive step, let $\sigma = \sigma'\, \ell$, and assume as induction hypothesis that $\sigma' \in \iutraces(\Delta(i)) \cap \oetraces(\Delta(s))$ (IH).
  
  We now establish that $\sigma' \in \utraces(\Delta(s))$ \;(3) holds, as follows:
  \begin{align*}
            & \sigma' \in \oetraces(s) \cup \iutraces(s) \ttag{Assumptions (2) and (IH)}\\
    \implies& \sigma' \in \traces(s) \cup \iutraces(s) \ttag{Lemma~\ref{lem:intersection-oe-iu-in-traces}}\\
    \implies& \sigma' \in \utraces(\Delta(s)) \ttag{Lemma~\ref{lem:utraces-iutraces}}
  \end{align*}
  Now, we distinguish two cases:
  \begin{itemize}
    \item 
      If $\ell \in I_i$, then $\sigma \in \oetraces(\Delta(s))$ holds by Definition~\ref{def:output-existential} of $\oetraces$.
      Furthermore,
      \begin{align*}
                & \sigma' \in \utraces(\Delta(s)) \wedge \sigma \in \iutraces(\Delta(s)) \ttag{Assumptions (2) and (3)}\\
        \implies& \sigma' \in \utraces(\Delta(s)) \wedge \ell \in \inp(\after{\Delta(s)}{\sigma})  \ttag{$\sigma = \sigma'\ell$ and Lemma~\ref{lem:oe-traces-form}}\\
        \implies& \ell \in \inp(\after{\Delta(i)}{\sigma'}) \ttag{Assumption (1)}\\
        \implies& \sigma \in \iutraces(\Delta(i)) \ttag{$\sigma = \sigma'\ell$}
      \end{align*}
    \item
      If $\ell \in O_i \cup \{ \delta \}$, then $\sigma \in \iutraces(\Delta(s))$ holds by Definition~\ref{def:input-universal} of $\iutraces$.
      Furthermore,
      \begin{align*}
                & \sigma' \in \utraces(\Delta(s)) \wedge \sigma \in \oetraces(\Delta(i))  \ttag{Assumptions (2) and (3)}\\
        \implies& \sigma' \in \utraces(\Delta(s)) \wedge \ell \in \out(\after{\Delta(i)}{\sigma})  \ttag{$\sigma = \sigma'\ell$ and Lemma~\ref{lem:oe-traces-form}}\\
        \implies& \ell \in \out(\after{\Delta(s)}{\sigma'}) \ttag{Assumption (1)}\\
        \implies& \sigma \in \oetraces(\Delta(s)) \ttag{$\sigma = \sigma'\ell$ and Lemma~\ref{lem:oe-traces-form}}
      \end{align*}
  \end{itemize}
  
  $(\impliedby)$
  Assume $\oetraces(\Delta(i)) \cap \iutraces(\Delta(s)) \subseteq \iutraces(\Delta(i)) \cap \oetraces(\Delta(s))$ (1).
  To prove $i \uioco s$, we assume some $\sigma \in \utraces(\Delta(s))$ (2), for which we will show  $\out(\after{\Delta(i)}{\sigma}) \subseteq \out(\after{\Delta(s)}{\sigma})$ and $\inp(\after{\Delta(i)}{\sigma}) \supseteq \inp(\after{\Delta(s)}{\sigma})$.
  Assumption (2) implies $\sigma \in \iutraces(\Delta(s)) \cap \oetraces(\Delta(s))$ (3) by Lemmas~\ref{lem:intersection-oe-iu-in-traces} and~\ref{lem:utraces-iutraces}.
  
  First, we establish that $\sigma \in \iutraces(\Delta(i))$ (4) holds, shown by induction to the length of $\sigma$.
  The base case $\sigma = \epsilon$ trivially holds, and for the inductive step, let $\sigma = \sigma'\ell$ and assume as induction hypothesis that $\sigma' \in \iutraces(\Delta(i))$ (IH).
  Now if $\ell$ is an output, the proof is trivial, so assume $\ell$ is an input.
  Then we distinguish two cases: $\sigma \in \oetraces(\Delta(i))$ or $\sigma \not\in \oetraces(\Delta(i))$.
  In the former case, $\sigma \in \iutraces(\Delta(i))$ follows from assumptions (1) and (3).
  In the latter case, $\sigma' \not\in \oetraces(\Delta(i))$ also holds, and Lemma~\ref{lem:intersection-oe-iu-in-traces} then implies that $\sigma' \not\in \traces(\Delta(i))$, so then $\ell \in \inp(\after{\Delta(i)}{\sigma'})$ vacuously holds, and together with (IH) this implies $\sigma \in \iutraces(\Delta(i))$.
  
  We now prove $\out(\after{\Delta(i)}{\sigma}) \subseteq \out(\after{\Delta(s)}{\sigma})$, by assuming some $\ell \in \out(\after{\Delta(i)}{\sigma})$ (5) and proving $\ell \in \out(\after{\Delta(s)}{\sigma})$.
  Assumption (5) implies $\sigma\ell \in \oetraces(\Delta(i))$, and assumptions (3) and (5) implies $\sigma\ell \in \iutraces(\Delta(s))$, so then assumption (1) implies $\sigma\ell \in \oetraces(\Delta(s))$.
  This proves $\ell \in \out(\after{\Delta(s)}{\sigma})$.
  
  We also prove $\inp(\after{\Delta(i)}{\sigma}) \supseteq \inp(\after{\Delta(s)}{\sigma})$, by assuming some $\ell \in \inp(\after{\Delta(s)}{\sigma})$ (6) and proving $\ell \in \inp(\after{\Delta(i)}{\sigma})$.
  This holds vacuously if $\sigma \not \in \traces(\Delta(i))$, so assume $\sigma \in \traces(\Delta(i))$ holds.
  Then $\sigma \in \oetraces(\Delta(i))$ also holds by Lemma~\ref{lem:oe-traces-form}, and then also $\sigma\ell \in \oetraces(\Delta(i))$ (7) holds.
  Assumptions (3) and (6) imply $\sigma\ell \in \oetraces(\Delta(s))$, so together with (7), this implies $\sigma\ell \in \iutraces(\Delta(i))$.
  This proves $\ell \in \inp(\after{\Delta(i)}{\sigma})$.
  \qed
\end{proof}

\begin{example}
  \label{exa:ioco-not-ir}
  
  Consider $s_E \in \iots$ and $s_F \in \lts$
  with $I_E = I_F = \{a\}$ and  $O_E = O_F = \{x,y\}$
  in Fig.~\ref{fig:example-ioco-not-uioco},
  where quiescence has been explicitly added.
  Implementation $s_E$ is not $\ioco$-conformant to specification $s_F$:
  if we consider the trace $aa \in \traces(\Delta(s_F))$
  then $y \in \out(\after{\Delta(s_E)}{aa})$
  but $y \not\in \out(\after{\Delta(s_F)}{aa})$.

  However, trace $aa$ does not disprove $\uioco$-conformance,
  since it is not $s_F$-input-universal:
  $aa \notin \utraces(\Delta(s_F))$,
  since $a \not\in \inp(\after{\Delta(s_F)}{a})$.
  In fact, $s_E \uioco s_F$ holds, which we prove via
  Theorems~\ref{the:ir-iuoe} and \ref{the:uioco-ir}
  by showing that $\Delta(s_E) \iuoe \Delta(s_F)$.
  We first establish that
  $\oetraces(\Delta(s_E)) \cap \iutraces(\Delta(s_F)) = \delta^* + \delta^* a \delta^*$:
  extending any trace $\sigma$ in this set by an output $\ell$ other than $\delta$ causes
  $\sigma\ell \not\in \oetraces(\Delta(s_E))$, and extending it by an input
  $\ell$ causes $\sigma\ell \not\in \iutraces(\Delta(s_F))$.
  Clearly, any trace in $\delta^* + \delta^* a \delta^*$ is also in
  $\iutraces(\Delta(s_E))$ and in $\oetraces(\Delta(s_F))$,
  so $\Delta(s_E) \iuoe \Delta(s_F)$ holds.
  It follows that $\Delta(s_E) \ir \Delta(s_F)$
  and $s_E \uioco s_F$ hold.
  
  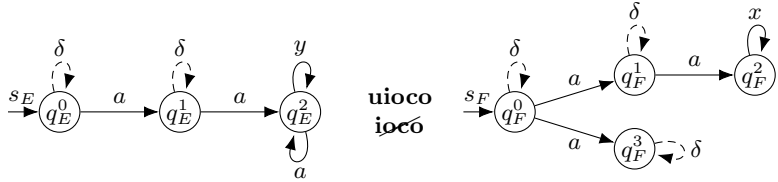
\begin{figure}[h!]
    \tikzlts
    \begin{center}
      \begin{tikzpicture}[node distance=16mm]
  \node [named,initial,initial text=] (0) {$q_E^0$};
  \node (init-text) [above left=-2mm and 0 of 0] {$s_E$};
  \node [named] (1) [right of=0] {$q_E^1$};
  \node [named] (2) [right of=1] {$q_E^2$};
  \path
  (0) edge [densely dashed, loop above] node {$\delta$} (0)
  (0) edge node [above] {$a$} (1)
  (1) edge [densely dashed, loop above] node {$\delta$} (1)
  (1) edge node [above] {$a$} (2)
  (2) edge [loop above] node {$y$} (2)
  (2) edge [loop below] node {$a$} (2)
  ;

  \node [align=center] [right=37mm of 0] {$\uioco$\\$\notrel{\ioco}$};
  \node [named,initial,initial text=] [right=55mm of 0] (0) {$q_F^0$};
  \node (init-text) [above left=-2mm and 0 of 0] {$s_F$};
  \node [named] (1a) [above right=1mm and 12mm of 0] {$q_F^1$};
  \node [named] (1b) [below right=1mm and 12mm of 0] {$q_F^3$};
  \node [named] (2) [right of=1a] {$q_F^2$};
  \path
  (0) edge [densely dashed, loop above] node {$\delta$} (0)
  (0) edge node [above] {$a$} (1a)
  (0) edge node [below] {$a$} (1b)
  (1a) edge [densely dashed, loop above] node {$\delta$} (1a)
  (1a) edge node [above] {$a$} (2)
  (1b) edge [densely dashed, loop right] node {$\delta$} (1b)
  (2) edge [loop above] node {$x$} (2)
  ;
      \end{tikzpicture}
    \end{center}
    \caption{$s_E \in \iots$ and $s_F \in \lts$.
             The dashed transitions are added by $\Delta$.}
    \label{fig:example-ioco-not-uioco}
  \end{figure}

\end{example}

\begin{longversion}
A last remark concerns the similarities and difference
between input-refusals (Def.~\ref{def:refusal} and~\ref{def:ir-inclusion})
and output refusals, or quiescence (Def.~\ref{def:ioco}.2).
Both are defined as refusals, i.e., some actions that can be refused
in some state, but
each input is treated separately, $a \not \in \inp(\after{s}{\sigma})$ for some $a$,
whereas outputs are only treated collectively, $x \not\in \out(\after{s}{\sigma})$ for all $x$.
Moreover, output refusals can occur anywhere in a trace
(cf.\ $\ldots$repetitive quiescence$\ldots$ \cite{Tre96}):
after quiescence a next input can occur.
Input refusal are final, i.e., they always occur as the last action
of a trace. 
As such, input refusals behave analogous to \emph{failures semantics}
in the \emph{linear time -- branching time spectrum} \cite{vG01},
whereas quiescence is analogous to \emph{failure-trace semantics}.
Relations where quiescence always occurs as last action in a trace
have also been defined e.g, \emph{quiescent-trace preorder}
in the context of I/O-Automata \cite{Vaa91}.
We might add \emph{repetitive input-refusals}, which would lead
stronger refinement relations, e.g., we would be able to discriminate
between $s_P$ and $s_Q$ in Fig.~\ref{fig:example-repetitive-input-refusal}:
$s_P$ and $s_Q$ are input-failure equivalent,
but the \emph{repetitive input-refusal trace}
$a\,\overline{b}\,\overline{c}$ would be able to tell them apart.

  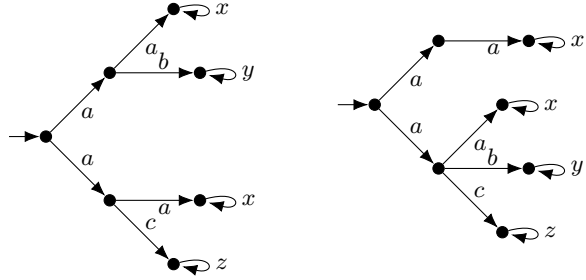
\begin{figure}[h!]
    \tikzlts
    \begin{center}
      $\vcenter{\hbox{\begin{tikzpicture}[node distance=12mm]
  \node [state,initial,initial text=] (0) {};
  \node (init-text) [above left=-2mm and 0 of 0] {};
  \node [state] (1a) [above right of=0]  {};
  \node [state] (1b) [below right of=0]  {};
  \node [state] (2a) [above right of=1a] {};
  \node [state] (2b) [right of=1a]       {};
  \node [state] (3a) [right of=1b]       {};
  \node [state] (3b) [below right of=1b] {};
  \path
  (0) edge node [below right=-1mm] {$a$} (1a)
  (0) edge node [above right=-1mm] {$a$} (1b)
  (1a) edge node [below right=-1mm] {$a$} (2a)
  (1a) edge node [above right=-1mm] {$b$} (2b)
  (1b) edge node [below right=-1mm] {$a$} (3a)
  (1b) edge node [above right=-1mm] {$c$} (3b)
  (2a) edge [loop right] node {$x$} (2a)
  (2b) edge [loop right] node {$y$} (2b)
  (3a) edge [loop right] node {$x$} (3a)
  (3b) edge [loop right] node {$z$} (3b)
  ;
      \end{tikzpicture}}}$
      \qquad
      $\vcenter{\hbox{\begin{tikzpicture}[node distance=12mm]
  \node [state,initial,initial text=] (0) {};
  \node (init-text) [above left=-2mm and 0 of 0] {};
  \node [state] (1a) [above right of=0]  {};
  \node [state] (1b) [below right of=0]  {};
  \node [state] (2)  [right of=1a]       {};
  \node [state] (3a) [above right of=1b] {};
  \node [state] (3b) [right of=1b]       {};
  \node [state] (3c) [below right of=1b] {};
  \path
  (0)  edge node [below right=-1mm] {$a$} (1a)
  (0)  edge node [above right=-1mm] {$a$} (1b)
  (1a) edge node [below right=-1mm] {$a$} (2)
  (1b) edge node [below right=-1mm] {$a$} (3a)
  (1b) edge node [above right=-1mm] {$b$} (3b)
  (1b) edge node [above right=-1mm] {$c$} (3c)
  (2)  edge [loop right] node {$x$} (2)
  (3a) edge [loop right] node {$x$} (3a)
  (3b) edge [loop right] node {$y$} (3b)
  (3c) edge [loop right] node {$z$} (3c)
  ;
      \end{tikzpicture}}}$
    \end{center}
    \caption{IA $s_P$ and IA $s_Q$.}
    \label{fig:example-repetitive-input-refusal}
  \end{figure}
  
\end{longversion}

\section{Game Characterizations}
\label{sec:alternating-trace-inclusion}

Ordinary trace containment can be seen as a game between a protagonist and antagonist: the antagonist chooses a path in the left-hand model, and the protagonist should find a path in the right-hand model having the same trace. 
Trace containment then holds if the protagonist can always win.
Alur et al.~\cite{Aluetal98} generalized this game to alternating-trace containment, which we will now compare to input-failure refinement, $\ioco$ and $\uioco$.

Alternating-trace containment acts on alternating transition systems.
Such a model is parameterized by a set of \emph{agents}, which are either collaborative or adversarial.
Every agent can restrict the possible transitions by choosing a \emph{strategy}.
If every agent has chosen a strategy, this yields a unique path following these choices.
The game of alternating-trace containment on models $s_1$ and $s_2$ is then played as follows.
First, the antagonist chooses a strategy for the collaborative agents in $s_1$.
Second, the protagonist chooses a matching strategy for the collaborative agents in $s_2$.
Third, the antagonist chooses a strategy for the adversarial agents in $s_2$, and fourth, the protagonist matches this choice for the adversarial agents in $s_1$.
In this way, the protagonist must ensure that the path in $s_1$ following these strategies has the same trace as the path in $s_2$.
Again, $s_1$ is alternating-trace contained in $s_2$ if the protagonist can always win.

\subsection{Alternating-Trace Containment for IA}

The agents in~\cite{Aluetal98} have no predefined roles, and any number of them may be defined.
In our setting, we instantiate a fixed number of agents to reflect the input-output-behaviour of a software system.
In particular, we introduce agents controlling the respective inputs and outputs, similarly to~\cite{BoSt18,AlHe01}.
In practice, a system itself acts as an agent controlling its outputs, whereas the environment serves as an agent controlling the inputs of the system.
The system and environment may also abstain from performing an action.


\begin{definition}
  \label{def:env-sys-strategies}
  \label{def:action-strategies}
  Let $s \in \lts$.
  An \emph{output strategy} for $s$ is a partial function $\stratsys : \Path(s) \rightharpoonup O$, 
  such that $\stratsys(\pi)\downarrow$ implies $\stratsys(\pi) \in \out(\last(\pi))$ for all $\pi$ (where $\stratsys(\pi)\downarrow$ means that $\stratsys(\pi)$ is defined).
  An \emph{input strategy} for $s$ is a partial function $\stratenvati : \Path(s) \rightharpoonup I$, 
  such that $\stratenvati(\pi)\downarrow$ implies $\stratenvati(\pi) \in \inp(\last(\pi))$.
  The domains of output and input strategies for $s$ are $\Stratsys(s)$ and $\Stratenvati(s)$ respectively.
\end{definition}

A system cannot only choose which outputs it produces, but also which transition it takes for a given input or output, in the case of non-determinism.
It also chooses how to resolve race conditions, that is, whether to take an input or an output transition, if both the input and output strategy choose an action.
To this end we introduce a \emph{determinization strategy} and a \emph{race condition strategy}.

\begin{definition}
  \label{def:det-strategy}
  Let $s \in \lts$.
  A \emph{determinization strategy} for $s$ is a partial function $\stratdet: \Path(s) \times L \rightharpoonup Q_s$ satisfying:
  (a) $q^n \xrightarrow{\ell}$ implies $\stratdet(q^0\ell^0\cdots q^n, \ell) \downarrow$, and
  (b) $\stratdet(q^0\ell^0\cdots q^n, \ell) = q^{n+1}$ implies $q^n \xrightarrow{\ell} q^{n+1}$.
  A \emph{race condition strategy} for $s$ is a function $\stratrc: \Path(s) \rightarrow \{0,1\}$, where $0$ denotes choosing the input in
  case of a race, whereas $1$ denotes choosing the output.
  The respective domains of determinization and race condition strategies for $s$ are denoted $\Stratdet(s)$ and $\Stratrc(s)$.
\end{definition}

\begin{longversion}
Note that for deterministic interface automata only a single, trivial determinisation strategy exists, so then $|\Stratdet(s)| = 1$.
\end{longversion}

The combination of an input strategy, an output strategy, a determinization strategy and a race condition strategy uniquely determines a path through an interface automaton.

\begin{definition} 
  \label{def:outcomes}
  Let $\strat = \langle \stratenvati, \stratsys, \stratdet, \stratrc \rangle \in \Stratenvati(s) \times \Stratsys(s) \times \Stratdet(s) \times \Stratrc(s)$ for $s \in \lts$, and let function $\operatorname{next}_{s, \strat} : \Path(s) \rightarrow \Path(s)$ be given by \[\operatorname{next}_{s, f}(\pi) =
  \begin{cases}
    \pi  \stratenvati(\pi) \stratdet(\pi,\stratenvati(\pi)) & \mbox{if } \stratenvati(\pi)\!\downarrow \wedge (\stratsys(\pi) \!\downarrow \Rightarrow \stratrc(\pi) = 0 )\\
   \pi  \stratsys(\pi)  \stratdet(\pi,\stratsys(\pi))& \mbox{if } \stratsys(\pi)\!\downarrow \wedge (\stratenvati(\pi) \!\downarrow \Rightarrow \stratrc(\pi) = 1 )\\
  \pi & \mbox{otherwise}
  \end{cases}
  \]
  Note that the infinite sequence
  $\pi_0, \pi_1,\ldots$ with
  $\pi_0 = q_s^0$ and
  $\forall j>0 :\pi_j = \operatorname{next}_{s, \strat}(\pi_{j-1})$ forms a chain of finite paths ordered by prefix.
  The \emph{outcome} of $s$ and $f$, notation $\Outc_s(f)$, is the limit under prefix ordering of $\pi_0, \pi_1,\ldots$.
  Observe that $\Outc_s(f)$ is either a finite path $\pi$ with $\operatorname{next}_{s, f}(\pi) = \pi$, or an infinite path.
\end{definition}

A software system is assumed to control its own outputs, as well as non-determinism as race conditions, so the corresponding strategies are collaborative.
Inputs are chosen by the environment, so the input strategy is adversarial.
This leads to the following instantiation of alternating-trace containment for IA.

\begin{definition}
  Let $s_1, s_2 \in \lts$.
  Then $s_1$ is \emph{alternating-trace contained} in $s_2$, denoted $s_1 \ati s_2$, if
  \begin{align*}
    &\forall \stratsys^1 \in \Stratsys(s_1), \forall \stratdet^1 \in \Stratdet(s_1), \forall \stratrc^1 \in \Stratrc(s_1),\\
    &\exists \stratsys^2 \in \Stratsys(s_2), \exists \stratdet^2 \in \Stratdet(s_2), \exists \stratrc^2 \in \Stratrc(s_2),\\
    &\forall \stratenvati^2 \in \Stratenvati(s_2), \exists \stratenvati^1 \in \Stratenvati(s_1):
    \\&\;
    \trace(\Outcati[1]) = \trace(\Outcati[2])
    \end{align*}
\end{definition}

Having defined alternating-trace containment for IA, we can now disprove the conjecture in~\cite{BoSt18}:
Alternating-trace containment does not coincide with $\ioco$, nor with $\uioco$, $\ir$ or $\iuoe$, as shown by Example~\ref{exa:problems-ati-with-quantifier-order}.

\begin{proof}
  Example~\ref{exa:problems-ati-with-quantifier-order} shows that $s_G$ is related to $s_H$ by $\iuoe$, $\ir$, $\ioco$ and $\uioco$, but not by $\ati$.
\end{proof}

\begin{example}
  \label{exa:problems-ati-with-quantifier-order}
  Consider IA $s_G$ and $s_H$ in Figure~\ref{fig:problems-ati-with-quantifier-order}.
  IA $s_G$ is input-enabled, so $\ioco$ can be applied.
  Both IA have an output transition in every state, so $\Delta$ has no effect, which implies that $\ir$ and $\uioco$ coincide, even without explicitly applying $\Delta$.
  All traces of $s_H$ are input-universal, so $\uioco$ and $\ioco$ also coincide.
  
  Then $\oetraces(s_G) \cap \iutraces(s_H)$ are the traces in $z^*az^*ax^*$ and $z^*az^*by^*$, and all prefixes of those traces.
  These are included in $\iutraces(s_G) \cap \oetraces(s_H)$, so $s_G \iuoe s_H$ holds, and $s_G$ is thus also related to $s_H$ by relations $\ir$, $\uioco$ and $\ioco$.
  
\begin{shortversion}
  Now, let us play the game of alternating-trace containment.
  The antagonist chooses outputs $x$ and $y$ in states $q_G^2$ and $q_G^{2\prime}$ respectively, and resolves race conditions such that states $q_G^2$ or $q_G^{2\prime}$ are reached.
  In this way, the antagonist could enforce traces $aax$ and $aby$, so the protagonist must match this with outputs $x$ and $y$ in respectively $q_H^2$ and $q_H^{2\prime}$.
  The protagonist should also choose a determinization strategy.
  Only two choices are possible: from $q_H^0$ it can make a transition to either $q_H^1$ or to $q_H^{1\prime}$.
  Suppose that the protagonist chooses $q_H^1$.
  The antagonist must then choose an input strategy, and it chooses to pick input $a$ after path $q_H^0$, and input $b$ after any path with trace $a$.
  Now, all strategies for $s_H$ have been chosen and the outcome is $q_H^0 a q_H^1 b q_H^3$, with trace $ab$.
  The protagonist should choose a matching strategy to pick inputs in $s_G$, but it cannot: if it picks inputs to produce trace $ab$, then the resulting trace in $s_G$ will be $aby$ instead of $ab$.
  The protagonist thus loses the game.
  Had the protagonist chosen a determininization strategy to $q_H^{1\prime}$, then he would lose in the same manner.
\end{shortversion}
\begin{longversion}
  Now, let us play the game of alternating-trace containment.
  The antagonist chooses a strategy which picks output $x$ after path $q_G^0\,aq_G^1\,aq_G^2$, output $y$ after $q_G^0\,aq_G^1\,bq_G^{2\prime}$, and no output otherwise.
  It resolves race conditions in $s_G$ by always choosing inputs in states $q_G^0$ and $q_G^1$.
  Since $s_G$ is deterministic, no determinization strategy needs to be chosen.
  
  The protagonist should now choose an output strategy.
  It never chooses $z$, since $z$ is also never chosen by the antagonist.
  Suppose the protagonist does not choose $x$ after $q_H^0\,aq_H^1\,aq_H^2$, then the protagonist would lose: the antagonist can then pick inputs following trace $aa$.
  This would unavoidably lead to trace $aa$ in $s_H$, but the protagonist cannot match this trace in $s_G$: it should then also pick inputs in $s_G$ following trace $aa$, but this would result in an outcome with trace $aax$.
  Thus, the protagonist should choose an output strategy that picks $x$ after any path with trace $aa$.
  However, it can choose $x$ only after path $q_H^0 a q_H^1 a q_H^2$, since this is the only path after which $x$ is enabled.
  In the same manner, it should also pick $y$ after $q_H^0 a q_H^{1\prime} b q_H^{2\prime}$.

  Furthermore, the protagonist should produce a determinization strategy.
  Only two choices are possible: from $q_H^0$, it can make a transition to either $q_H^1$ or to $q_H^{1\prime}$.
  Suppose that the protagonist chooses $q_H^1$.
  The antagonist must then choose an input strategy, and it chooses one which picks input $a$ after path $q_H^0$, and input $b$ after any path with trace $a$.
  Now, all strategies for $s_H$ have been chosen: they follow the path $q_H^0 a q_H^1 b q_H^4$, so they produce trace $ab$.
  The protagonist should then choose a matching strategy to pick inputs in $s_G$, but it cannot: it should pick at least input $a$ after path $q_G^0$ and input $b$ after path $q_G^0 a q_G^1$ to match trace inputs $a$ and $b$, but then the strategies for $s_G$ follow path $q_G^0 a q_G^1 b q_G^{2\prime} y q_G^{3\prime}$, which produces trace $aby$.
  The protagonist has thus lost the game.
  Had the protagonist chosen a different determinization strategy to state $q_H^{1\prime}$, then it would have lost in the same manner, so no winning strategy exists.
\end{longversion}

  Thus, $s_G \not\ati s_H$ holds, even though $s_G$ is related to $s_H$ by relations $\ir$, $\iuoe$, $\uioco$ and $\ioco$.
  The intuitive reason is that the protagonist must already choose a determinization strategy for $s_H$, before the antagonist chooses an input strategy for $s_H$.
  Would this order of turns be reversed, then the protagonist could win the game for this example.
  The protagonist could then choose the determinization strategy for $s_H$ such that either trace $aax$ or trace $aby$ is matched, depending on which inputs are chosen by the antagonist in $s_G$.
  
  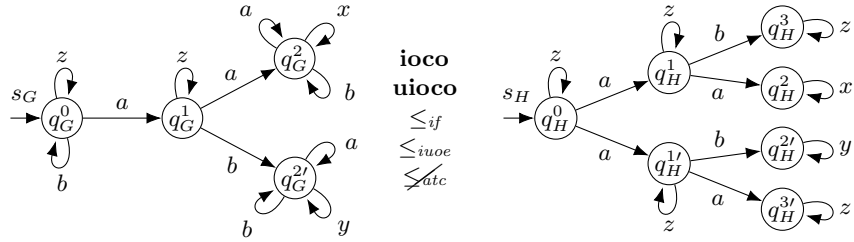
\begin{figure}
    \tikzlts
    \begin{center}
      \begin{tikzpicture}[node distance=16mm]
	\node [named,initial,initial text=] (0) {$q_G^0$};
	\node (init-text) [above left=-1mm and 0 of 0] {$s_G$};
	\node [named] (1) [right of=0] {$q_G^1$};
	\node [named] (2) [above right=4mm and 11mm of 1] {$q_G^2$};
	\node [named] (2p) [below right=4mm and 11mm of 1] {$q_G^{2\prime}$};
	\path
  (0) edge [loop below] node {$b$} (0)
	(0) edge node [above] {$a$} (1)
	(1) edge node [above left=0mm and -1mm] {$a$} (2)
	(1) edge node [below left=0mm and -1mm] {$b$} (2p)
	(2) edge [in=30,out=60,loop] node [above right] {$x$} (2)
  (2) edge [in=120,out=150,loop] node [above left] {$a$} (2)
  (2) edge [in=-60,out=-30,loop] node [right=1mm] {$b$} (2)
	(2p) edge [in=-60,out=-30,loop] node [below right] {$y$} (2p)
  (2p) edge [in=30,out=60,loop] node [right=1mm] {$a$} (2p)
  (2p) edge [in=-150,out=-120,loop] node [below left] {$b$} (2p)
  (0) edge [loop above] node {$z$} (0)
  (1) edge [loop above] node {$z$} (1)
	;
	\node [align=center] [right=40mm of 0] {$\ioco$\\$\uioco$\\$\ir$\\$\iuoe$\\$\notrel{\ati}$};
	
	\node [named,initial,initial text=] [right=60mm of 0] (0) {$q_H^0$};
	\node (init-text) [above left=-1mm and 0 of 0] {$s_H$};
	\node [named] (1) [above right=2mm and 11mm of 0] {$q_H^1$};
	\node [named] (2) [below right=-2mm and 11mm of 1] {$q_H^2$};
	\node [named] (3) [above right=2mm and 11mm of 1] {$q_H^3$}; 
	\node [named] (1p) [below right=2mm and 11mm of 0] {$q_H^{1\prime}$};
	\node [named] (2p) [above right=-2mm and 11mm of 1p] {$q_H^{2\prime}$};
	\node [named] (3p) [below right=2mm and 11mm of 1p] {$q_H^{3\prime}$};
	\path
	(0) edge node [above left=0mm and -1mm] {$a$} (1)
	(0) edge node [below left=0mm and -1mm] {$a$} (1p)
	(1) edge node [below left=0mm and -1mm] {$a$} (2)
	(1) edge node [above left=0mm and -1mm] {$b$} (3)
	(1p) edge node [above left=0mm and -1mm] {$b$} (2p)
	(2) edge [loop right] node {$x$} (2)
	(2p) edge [loop right] node {$y$} (2p)
	(1p) edge node [below left=0mm and -1mm] {$a$} (3p)
	
  (0) edge [loop above] node {$z$} (0)
  (1) edge [loop above] node {$z$} (1)
  (3) edge [loop right] node {$z$} (3)
  (1p) edge [loop below] node {$z$} (1p)
  (3p) edge [loop right] node {$z$} (3p)
	;
      \end{tikzpicture}
    \end{center}
    \caption{IA $s_G$ and $s_H$.}
    \label{fig:problems-ati-with-quantifier-order}
  \end{figure}
\end{example}

The analysis of $s_G \not\ati s_H$ in Example~\ref{exa:problems-ati-with-quantifier-order} is rather complex.
An intuitive experiment showing the difference between $s_G$ and $s_H$ would improve understanding of $\ati$, but unfortunately, no observational interpretation of alternating-trace containment is given in~\cite{Aluetal98}.
\begin{longversion}
Clearly, experiments characterizing (alternating) simulation~\cite{Abr87} suffice, but we will show in Section~\ref{sec:alternating-simulation} that alternating simulation is stronger than alternating-trace containment.
Therefore, such experiments are too strong: they distinguish IA for which alternating-trace containment holds.
\end{longversion}

For non-input-enabled IA, another difference between alternating-trace containment and the other relations is shown in Example~\ref{exa:problems-ati-trace-based}.

\begin{example}
  \label{exa:problems-ati-trace-based}
  Consider IA $s_I$ and $s_J$ in Figure~\ref{fig:problems-ati-trace-based}.
  Clearly, $s_I \iuoe s_J$ holds, since $\oetraces(s_I) \cap \iutraces(s_J)$ are the traces $y^* + y^*xy^*$ and their prefixes, which are in $\iutraces(s_I)$ and $\oetraces(s_J)$.
  Therefore, $s_I \ir s_J$ and $s_I \uioco s_J$ hold as well by the same resoning as in Example~\ref{exa:problems-ati-with-quantifier-order}.
  Since $s_I$ is not input-enabled, $\ioco$ is not defined.

\begin{shortversion}
  Now, we play the game of alternating-trace containment.
  In $s_I$, the antagonist first picks output $x$, and then abstains performing an output, so the trace of the outcome in $s_I$ is $x$.
  Clearly, the protagonist can match this in $s_J$.
  But no matter the determinization strategy he chooses, if the antagonist picks inputs $a$ and $b$ in $s_J$ whenever possible, then the protagonist cannot match these inputs in $s_I$.
\end{shortversion}
\begin{longversion}
  Now, we play the game of alternating-trace containment.
  The antagonist picks an output strategy with $\stratsys^G(q_I^0) = x$ and $\stratsys^G(q_I^0 x q_I^1) = \bot$, so it follows that the trace of the outcome in $s_I$ is $x$.
  
  Now, the protagonist must choose a determinization strategy, an output strategy and a race condition strategy for $s_J$.
  To match trace $x$, the output strategy must have $\stratsys^G(q_J^0) = x$.
  Suppose the determinization strategy has $\stratdet^G(q_J^0, x) = q_J^1$.
  Then to ensure that $xy$ is not an outcome, $\stratsys^G(q_I^0 x q_J^1) = \bot$ must be chosen.
  But then the antagonist can choose $\stratenvati^G(q_J^0 x q_J^1) = a$, resulting in an outcome with trace $xa$.
  Likewise, if the determinization strategy chooses the lower branch to $q_J^{1\prime}$, the antagonist can enforce trace $xb$.
  Traces $xa$ and $xb$ cannot be matched in $s_I$, so the protagonist loses and $s_I \not\ati s_J$ holds.
\end{longversion}

  The intuitive reasoning is that the antagonist may choose inputs $a$ and $b$ after paths $q_J^0 x q_J^1$ and $q_J^0 x q_J^{1\prime}$, respectively, whereas $a$ and $b$ are not universally enabled after trace $x$.
  Would the antagonist pick only inputs in $\inp(\after{s_J}{x})$, then the protagonist could win the game.

  \begin{figure}
    \tikzlts
    \begin{center}
      \begin{tikzpicture}[node distance=16mm]
        \node (center) {};
        \node [named,initial,initial text=] (0) {$q_I^0$};
        \node (init-text) [above left=-2mm and 0 of 0] {$s_I$};
        \node [named] (1) [right of=0] {$q_I^1$};
        \path
        (0) edge node [above] {$x$} (1)
        (0) edge [loop above] node {$y$} (0)
        (1) edge [loop above] node {$y$} (1)
        ;
        
        \node [align=center] [right=25mm of center] {$\uioco$\\$\ir$\\$\iuoe$\\$\notrel{\ati}$};
        
        \node [named,initial,initial text=] [right=45mm of center] (0) {$q_J^0$};
        \node (init-text) [above left=-2mm and 0 of 0] {$s_J$};
        \node [named] (1) [above right of=0] {$q_J^1$};
        \node [named] (2) [right of=1] {$q_J^2$};
        \node [named] (1p) [below right of=0] {$q_J^{1\prime}$};
        \node [named] (2p) [right of=1p] {$q_J^{2\prime}$};
        \path
        (0) edge node [above left] {$x$} (1)
        (1) edge node [below] {$a$} (2)
        (0) edge node [below left] {$x$} (1p)
        (1p) edge node [above] {$b$} (2p)
        
        (0) edge [loop above] node {$y$} (0)
        (1) edge [loop below] node {$y$} (1)
        (2) edge [loop below] node {$y$} (2)
        (1p) edge [loop above] node {$y$} (1p)
        (2p) edge [loop above] node {$y$} (2p)
        ;
      \end{tikzpicture}
    \end{center}
    \caption{Interface automata $s_I$ and $s_J$.}
    \label{fig:problems-ati-trace-based}
  \end{figure}
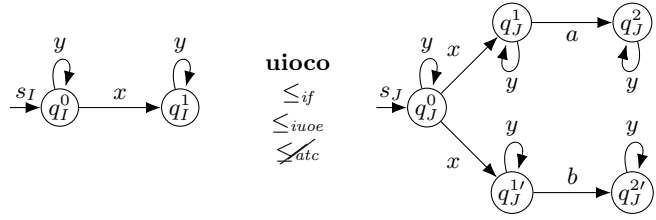
\end{example}

\subsection{The Game of Input-Failure Refinement}

Based on Examples~\ref{exa:problems-ati-with-quantifier-order} and~\ref{exa:problems-ati-trace-based}, we change the rules of the game of alternating-trace containment, in order to obtain a slightly weaker relation with a clearer observational meaning.
First, we argue that an environment usually cannot observe the precise state of a system, and thus also not the path taken by the system.
It can only observe traces of inputs and outputs, which restricts the input strategies.

\begin{definition}
  \label{def:trace-based}
  For $s \in \lts$, an input strategy $\stratenv$ is \emph{trace-based} if, for all $\pi_1,\pi_2 \in \Path(s)$, $\trace(\pi_1) = \trace(\pi_2)$ implies $\stratenv(\pi_1) = \stratenv(\pi_2)$.
  The domain of trace-based input strategies for $s$ is denoted $\Stratenv(s)$.
\end{definition}

\begin{longversion}
\begin{fact}
  \label{fac:env-strat-universal}
  Let $s \in \lts$, $\stratenv \in \Stratenv(s)$ and $\pi \in \Path(s)$.
  Then $\stratenv(\pi)\downarrow$ implies $\stratenv(\pi) \in \inp(\after{s}{\trace(\pi)})$.
\end{fact}
\end{longversion}

A second change is the order of turns.
The antagonist must first resolve all its choices, before the protagonist resolves any choices.

\begin{definition}
  \label{def:ati}
  Let $s_1, s_2 \in \lts$.
  Then $s_1 \aaee s_2$, if
  \begin{align*}
    &\forall \stratenv^2 \in \Stratenv(s_2), \forall \stratdet^1 \in \Stratdet(s_1), \forall \stratsys^1 \in \Stratsys(s_1), \forall \stratrc^1 \in \Stratrc(s_1),\\
    &\exists \stratenv^1 \in \Stratenv(s_1), \exists \stratdet^2 \in \Stratdet(s_2), \exists \stratsys^2 \in \Stratsys(s_2), \exists \stratrc^2 \in \Stratrc(s_2):\\
    &~~~~ \trace(\Outc(\stratenv^1,\stratsys^1,\stratdet^1,\stratrc^1)) = \trace(\Outc(\stratenv^2,\stratsys^2,\stratdet^2,\stratrc^2))
  \end{align*}
\end{definition}

In contrast to alternating-trace containment, this game has a correspondence with $\ioco$ theory in the non-deterministic setting.
It does not coincide with $\ioco$, but with $\uioco$.
We show this in Theorem~\ref{the:iuoe-aaee}, via input-failure refinement.
A technical detail is that this correspondence only holds in both directions when the right-hand interface automaton is image-finite.

\begin{definition}
  Interface automaton $s$
  is \emph{image-finite} if, for each $q \in Q_s$ and $\ell \in L_s$, $q$ has finitely many $\ell$-successors, i.e.,
  set $\{ q' \mid (q,\ell,q') \in T_s \}$ is finite.
\end{definition}

\begin{theorem}
  \label{the:iuoe-aaee}
  Let $s_1, s_2 \in \lts$.
  Then\:
  $s_1 \aaee s_2 \implies s_1 \iuoe s_2$.\\
  Furthermore, if $s_2$ is image-finite, then\:
  $s_1 \aaee s_2 \impliedby s_1 \iuoe s_2$.
\end{theorem}

\begin{proof}
  $(\implies)$
  We prove the contrapositive: assume $s_1 \not\iuoe s_2$, and we prove $s_1 \not\aaee s_2$.
  By assumption $s_1 \not\iuoe s_2$, there exists a sequence $\sigma \in L^*$ with
  $\sigma\in\oetraces(s_1) \cap \iutraces(s_2)$ and $\sigma\not\in \iutraces(s_1) \cap \oetraces(s_2)$.
  We define strategy functions $\stratenv^2$, $\stratsys^1$, $\stratdet^1$ and $\stratrc^1$ that try to realize $\sigma$ as an outcome, whenever possible.
  Strategy $\stratenv^2$ is defined as follows, for  $\pi_2 \in\Path(s_2)$,
  \begin{align*}
  \stratenv^2(\pi_2) = & 
  \begin{cases}
  a &\text{if } \trace(\pi_2) \; a \text{ is a prefix of } \sigma \\
  \bot &\text{otherwise}
  \end{cases}\\
  \end{align*}
  Note that $\stratenv^2$ is a trace-based action strategy: since $\sigma\in \iutraces(s_2)$, input symbol $a$ is enabled after every path of $s_2$ with the same trace as $\pi_2$.
  We say that path $\pi_1 \in \Path(s_1)$ \emph{can realize $\sigma$ via $\ell \in L$} if there exists a path $\pi \in \Path(s_1)$ with $\trace(\pi) = \sigma$, $\pi_1$ a proper prefix of $\pi$, and $\ell$ the first symbol in $\pi$ following $\pi_1$.
  Now we define $\stratsys^1$ and $\stratrc^1$ as follows, for $\pi_1 \in\Path(s_1)$,
    \begin{align*}   
      \stratsys^1(\pi_1) = & 
        \begin{cases}
          x &\text{if } \pi_1 \text{ can realize } \sigma \text{ via } x \in O \\
          \bot &\text{otherwise}
        \end{cases}\\
          \stratrc^1(\pi_1) = & 
              \begin{cases}
              1 & \text{if } \pi_1 \text{ can realize } \sigma \text{ via some } x \in O \\
              0 & \text{otherwise}
              \end{cases}    
    \end{align*}
  Note that $\stratsys^1$ is an output strategy, because if $\pi_1$ can realize $\sigma$ via $x$, then $x$ is enabled in the last state of $\pi_1$.
  In addition, we choose determinization strategy $\stratdet^1$ such that
  \begin{eqnarray*}
  \stratdet^1(\pi_1,\ell) = q & \wedge &
  \pi_1 \text{ can realize } \sigma  \text{ via } \ell \\  
  & \implies & (\trace(\pi_1) \; \ell = \sigma) \vee (\pi_1 \; \ell \; q \text{ can realize } \sigma \text{ via some } \ell').
  \end{eqnarray*}   
  We claim that, no matter how we define $\stratenv^1$, $\stratsys^2$, $\stratdet^2$ and $\stratrc^2$,
  \begin{eqnarray*}
  \trace(\Outc[1]) & \neq & \trace(\Outc[2]).
  \end{eqnarray*}
  Since $\sigma\not\in \iutraces(s_1) \cap \oetraces(s_2)$, $\sigma$ is nonempty.
  Let $\ell$ be the first symbol occurring in $\sigma$. If $\ell \in I$, then the
  strategies for $s_2$ will do either $\ell$ or an output symbol to start $\trace(\Outc[2])$, whereas the strategies for $s_1$ will either choose an input symbol to start $\trace(\Outc[1])$, or choose to terminate so that $\trace(\Outc[1]) = \epsilon$.
  Thus, the only way in which both strategies end up with the same trace is by performing
  an $\ell$-step.
  Otherwise, if $\ell \in O$, then the strategy for $s_1$ will choose to do $\ell$, whereas
  the strategy for $s_2$ will either choose an output symbol or choose to terminate.
  Again, the only way in which both strategies end up with the same trace is by performing
  an $\ell$-step.
  By repeating the same argument, we see that the only way in which both strategies possibly may end up with the same trace is by selecting paths with trace $\sigma$.
  But this is not possible since  $\sigma\not\in \iutraces(s_1) \cap \oetraces(s_2)$:
  at some point either the strategy for $s_1$ will fail to match an input transition,
  or the strategy for $s_2$ will fail to match an output transition.
  Consequently, $\stratenv^2$, $\stratsys^1$, $\stratdet^1$ and $\stratrc^1$ are witnesses proving $s_1 \not\aaee s_2$.
  This proves that the contrapositive holds, that is, $s_1 \aaee s_2 \implies s_1 \iuoe s_2$. 

  $(\impliedby)$
  Assume $s_2$ is image-finite and $s_1 \iuoe s_2$ (1).
  We prove $s_1 \aaee s_2$.
  Let $\stratenv^2 \in \Stratenv(s_2)$, $\stratdet^1 \in \Stratdet(s_1)$, $\stratsys^1 \in \Stratsys(s_1)$ and $\stratrc^1 \in \Stratrc(s_1)$.
  
  First we define, for all $\pi_1 \in \Path(s_1)$,
  \begin{align*}
    \stratenv^1(\pi_1) &= 
    \begin{cases}
      \stratenv^2(\pi_2) &\text{if $\pi_2 \in \Path(s_2)$ and }\\&\qquad\text{$\trace(\pi_1) = \trace(\pi_2) \in \oetraces(s_1) \cap \iutraces(s_2)$}\\
      \bot &\text{otherwise}
    \end{cases}\\
  \end{align*}
  By definition of $\stratenv^1$, $\pi_1 \in \Path(s_1)$ and $\stratenv^1(\pi_1) \downarrow$ implies that there exists
  some $\pi_2 \in \Path(s_2)$ such that 
  \begin{align*}
            & \trace(\pi_1) = \trace(\pi_2) \in \oetraces(s_1) \cap \iutraces(s_2)\\
    \implies& \trace(\pi_1) = \trace(\pi_2)\text{ and }\trace(\pi_2)\stratenv^2(\pi_2) \in \oetraces(s_1) \cap \iutraces(s_2) \ttag{Fact~\ref{fac:env-strat-universal} and Definition~\ref{def:input-universal} of $\iutraces$ and $\oetraces$}\\
    \implies& \trace(\pi_1)\stratenv^1(\pi_1) \in \oetraces(s_1) \cap \iutraces(s_2) \ttag{construction of $\stratenv^1$}\\
    \implies& \trace(\pi_1)\stratenv^1(\pi_1) \in \iutraces(s_1) \cap \oetraces(s_2) \ttag{Assumption (1)}\\
    \implies& \stratenv^1(\pi_1) \in \inp(\last(\pi_1))  \quad \ttag{Definition~\ref{def:input-universal} of $\iutraces$}
  \end{align*}
  This means $\stratenv^1(\pi_1)$ meets the conditions for input strategies in Definition~\ref{def:action-strategies}.
  Clearly, $\stratenv^1$ is also trace-based, so $\stratenv^1 \in \Stratenv(s_1)$ holds.
  
  Let $\pi = \Outc(\stratenv^1, \stratsys^1, \stratdet^1, \stratrc^1)$.
  Now consider the following digraph $G = (V, E)$:
  \begin{eqnarray*}
  V & = & \{ \pi_2 \in\Path(s_2) \mid \exists \pi_1 :  \pi_1 \mbox{ prefix of } \pi \mbox{ with } \trace(\pi_1) = \trace(\pi_2) \},\\
  E & = & \{ (\pi_2, \pi_2') \in V \times V \mid \pi_2' \mbox{ extends } \pi_2 \mbox{ with a single transition } \}.
  \end{eqnarray*}
  Note that $V$ is a prefix-closed set of finite paths of $s_2$, that each vertex in $V$ has a finite outdegree (since $s_2$ is image-finite),
  and that digraph $G$ is a tree.
  Let $\pi_1$ be a finite prefix of $\pi$ with $\trace(\pi_1) = \sigma$.
  It follows from the definitions of $\Outc$ and $\stratenv^1$ that $\sigma \in \oetraces(s_1) \cap \iutraces(s_2)$.
  Hence, by assumption (1), $\sigma \in \iutraces(s_1) \cap \oetraces(s_2)$.
  From this we infer $\sigma \in \iutraces(s_2) \cap \oetraces(s_2)$ and
  (using Lemma~\ref{lem:intersection-oe-iu-in-traces}) $\sigma\in\traces(s_2)$.
  This means that, for any prefix $\pi_1$ of $\pi$, $V$ contains a path $\pi_2$ with $\trace(\pi_1) = \trace(\pi_2)$.
  In particular, if $\pi$ is finite then $V$ contains a path $\hat{\pi}$ with $\trace(\pi) = \trace(\hat{\pi})$.
  Moreover, if $\pi$ is infinite then, by K\"onigs infinity lemma~\cite{Knu97}, digraph $G$ has an infinite path from the root, which corresponds to
  an infinite path $\hat{\pi}$ of $s_2$ with $\trace(\pi) = \trace(\hat{\pi})$.

  Based on $\hat{\pi}$, we define $\stratsys^2$, $\stratdet^2$ and $\stratrc^2$ as follows, for all $\pi_2 \in \Path(s_2)$,
  \begin{align*}
    \stratsys^2(\pi_2) &=
      \begin{cases}
        x &\text{ if $\pi_2 x$ is a prefix of $\hat{\pi}$}\\
        \bot & \text{otherwise}
      \end{cases}\\
    \stratdet^2(\pi_2, \ell) &= 
      \begin{cases}
        q &\text{ if $\pi_2 \ell q$ is a prefix of $\hat{\pi}$}\\
        \text{arbitrary} & \text{otherwise}
      \end{cases}\\
    \stratrc^2(\pi_2) &=
      \begin{cases}
        0 & \text{ if $\pi_2 a$ is a prefix of $\hat{\pi}$, for some $a \in I$}\\
        1 & \text{ if $\pi_2 x$ is a prefix of $\hat{\pi}$, for some $x \in O$}\\
        \text{arbitrary} & \text{ otherwise}
      \end{cases}
  \end{align*}
  We claim that $\Outc(\stratenv^2, \stratsys^2, \stratdet^2, \stratrc^2) = \hat{\pi}$. The definitions of strategies $\stratsys^2$, $\stratdet^2$ and $\stratrc^2$ are all geared towards outcome $\hat{\pi}$. 
  But also $\stratenv^2$ steers the outcome towards $\hat{\pi}$.
  Because suppose $\pi_2 a$ is a prefix of $\hat{\pi}$, for some $a \in I$.
  Let $\trace(\pi_2) = \sigma$.
    Then $\sigma \in \iutraces(s_1) \cap \oetraces(s_2)$ and there exists a prefix $\pi_1 a$ of $\pi$ with $\trace(\pi_1) = \sigma$.
    This implies $\stratenv^1(\pi_1) = a$.
    Hence, by definition of $\stratenv^1$, $\stratenv^1(\pi_1) = \stratenv^2(\pi_2)$ and thus
    $\stratenv^2(\pi_2) = a$.
    Using this observation, allows us to prove $\Outc(\stratenv^2, \stratsys^2, \stratdet^2, \stratrc^2) = \hat{\pi}$ with a simple
    inductive argument.
  Hence 
  \begin{align*}          
    & \trace(\Outc(\stratenv^2, \stratsys^2, \stratdet^2, \stratrc^2))
    \hspace{-0.5mm}=\hspace{-0.5mm} \trace(\hat{\pi})
    \hspace{-0.5mm}=\hspace{-0.5mm} \trace(\pi) 
    \hspace{-0.5mm}=\hspace{-0.5mm} \trace(\Outc(\stratenv^1, \stratsys^1, \stratdet^1, \stratrc^1)),
  \end{align*}
  which implies $s_1 \aaee s_2$, as required.
  \qed
\end{proof}

\begin{example}
  We revisit Example~\ref{exa:ir} to investigate the game-characterization of input-failure refinement.
  The IA in Figure~\ref{fig:ir} are image-finite so we should find the same related IA.
\begin{shortversion}
  First, consider $s_B \not\aaee s_A$.
  The antagonist chooses $\stratenv^A$, $\stratsys^B$ and $\stratrc^B$ by following the trace $axax$ in both models.
  Strategy $\stratdet^B$ is fixed, since $s_B$ is deterministic.
  Now, the protagonist should choose strategies so that the traces of the resulting outcomes match.
  It can match the first three actions $axa$, but this leads to $q_A^2$ in $s_A$.
  Here the protagonist cannot match the last action $x$, so indeed $s_B \not\aaee s_A$ holds.
  We show $s_C \not\aaee s_A$ similarly.
  Let the antagonist choose strategies following trace $axa$.
  Now, the protagonist ends up in state $q_A^1$ and $q_C^1$ after two actions, where it cannot match the third action $a$ in $s_C$.
  Finally, $s_D \aaee s_A$ holds.
  Since $q_D^0$ and $q_D^1$ do not have outgoing output transitions and $q_A^2$ has no input transitions, the antagonist can only choose a single action $a$ in $s_A$, which can easily be matched by the protagonist $s_D$.
\end{shortversion}
\begin{longversion}
  First, consider $s_B \not\aaee s_A$.
  The antagonist must first choose $\stratenv^A$, $\stratsys^B$, $\stratdet^B$ and $\stratrc^B$.
  It tries to follow the trace $axax$ in both models.
  That is, it chooses
  \begin{align*}
    \stratenv^A(q_A^0) &= a
      & \stratsys^B(q_B^0) &= \bot\\
    \stratenv^A(q_A^0 a q_A^1) &= \bot
      & \stratsys^B(q_B^0 a q_B^0) &= x\\
    \stratenv^A(q_A^0 a q_A^1 x q_A^1) &= a
      & \stratsys^B(q_B^0 a q_B^0 x q_B^0) &= \bot\\
    \stratenv^A(q_A^0 a q_A^1 x q_A^1 a q_A^2) &= \bot
      & \stratsys^B(q_B^0 a q_B^0 x q_B^0 a q_B^0) &= x\\
    \stratrc^B(q_B^0) &= 0\\
    \stratrc^B(q_B^0 a q_B^0) &= 1\\
    \stratrc^B(q_B^0 a q_B^0 x q_B^0) &= 0\\
    \stratrc^B(q_B^0 a q_B^0 x q_B^0 a q_B^0) &= 1
  \end{align*}

  Now, the protagonist should choose $\stratenv^B$, $\stratsys^A$, $\stratdet^A$ and $\stratrc^A$ such that the traces of the resulting outcomes for $s_A$ and $s_B$ match.
  
  Clearly, $x \not\in \trace(\Outc[B])$ because of the choice of $\stratsys^B$ by the antagonist.
  Thus, the protagonist must choose $\stratsys^A(q_A^0) \neq x$ or $\stratrc^A(q_A^0) = 0$ to match traces.
  In both cases, $a \in \trace(\Outc[A])$ holds, so the protagonist must match this with $\stratenv^B(q_B^0) = a$.
  By the choice of $\stratsys^B$ by the antagonist, this causes $a \in \trace(\Outc[B])$.
  The protagonist also needs to resolve non-determinism in $s_A$: trace $a$ leads to either $q_A^1$ or $q_A^2$.
  Choosing the latter state makes the protagonist lose directly, since it will then fail to match the output $x$ of $q_B^0$ in state $q_A^2$.
  Choosing the former state, we follow the same line of reasoning of defining strategies step by step.
  We eventually conclude that the protagonist is forced to choose its strategies such that $q_A^0 a q_A^1 x q_A^1 a q_A^2 \in \Outc[A]$ and $q_B^0 a q_B^0 x q_B^0 a q_B^0 \in \Outc[B]$.
  Here, the protagonist loses the game: the antagonist chooses $\stratsys^B(q_B^0 a q_B^0 x q_B^0 a q_B^0) = x$, which cannot be matched in state $q_A^2$.
  Thus, $s_B \not\aaee s_A$ indeed holds.
  
  Next, $s_C \not\aaee s_A$ can be shown by a similar approach.
  The antagonist chooses strategies following trace $axa$.
  The protagonist matches the first action $a$ by choosing $\stratenv^C(q_C^0) = a$ and $\stratdet^A(q_A^0, a) = q_A^1$, and the second action $x$ by choosing $\stratsys^A(q_A^0 a q_A^1) = x$.
  However, it cannot match the third action $a$: this would require choosing $\stratenv^C(q_C^0 a q_C^1 x q_C^1) = a$, but this is impossible, since $a$ is not enabled in $q_C^1$.
  This confirms $s_C \not\aaee s_A$.
  
  Now, we play a similar game for $s_D \aaee s_A$.
  Since $q_D^0$ and $q_D^1$ do not have outgoing output transitions and $q_A^2$ has no input transitions, the antagonist must choose $\stratsys^B(q_D^0) = \bot$, $\stratsys^B(q_D^0 a q_D^1) = \bot$ and $\stratenv^A(q_A^0 a q_A^2) = \bot$.
  If it would also choose $\stratenv^A(q_A^0) = \bot$, then the protagonist could win in a trivial way by always choosing $\bot$ as well, so the antagonist chooses $\stratenv^A(q_A^0) = a$.
  The protagonist matches this as follows:
  \begin{align*}
    \stratenv^D(q_D^0) &= a & \stratenv^D(q_D^0 a q_D^1) &= \bot\\
    \stratsys^A(q_A^0) &= \bot & \stratsys^D(q_A^0 a q_A^2) &= \bot\\
    \stratdet^A(q_A^0, a) &= q_A^2 
  \end{align*}
  Clearly, $\trace(\Outc[A]) = \trace(\Outc[B]) = \{\epsilon, a\}$ now holds.
  Since the antagonist could not have played the game differently, this is a winning strategy for the protagonist, proving $s_D \aaee s_A$.
\end{longversion}
\end{example}

\begin{shortversion}
  For image-infinite IA, the relation between $\aaee$ and $\iuoe$ is indeed strict, as is proven in~\cite{JaVaTr19}.
  For completeness, we also establish that $\aaee$ is indeed weaker than $\ati$.
  Strictness follows from Examples~\ref{exa:problems-ati-with-quantifier-order} and~\ref{exa:problems-ati-trace-based}.
\end{shortversion}

\begin{longversion}
For image-infinite IA, Theorem~\ref{the:iuoe-aaee} states that the game-characterization is stronger than input-failure refinement.
Example~\ref{exa:image-finite} shows that this implication is then indeed strict.
\begin{example}
  \label{exa:image-finite}
  Consider IA $s_K$ and $s_L$ in Figure~\ref{fig:problems-ati-image-finite}, where $s_L$ is infinitely branching: there is an infinite number of paths from the initial state, but each path has a finite length.
  Any positive integer $n$ thus has $\traces(q_L^n) = x^n$.
  Consequently, $\traces(s_L) = x^*$.
  Moreover, $\traces(s_K) = x^*$ holds as well.
  IA $s_K$ and $s_L$ are thus trace-equivalent, and since no inputs are present, they are also input-failure equivalent.
  Likewise, $s_K$ is also related to $s_L$ by $\iuoe$, $\ioco$ and $\uioco$.

  In the game of $\aaee$, the antagonist first picks a strategy to choose output transitions in $s_K$.
  Suppose it chooses the transition $q_K^0 \xrightarrow{x} q_K^0$ indefinitely.
  The trace of the outcome in $s_K$ following this strategy is thus the infinite trace $x x x\dots$.
  To match these traces, the protagonist should pick an output strategy which also keep producing output $x$.
  Furthermore, the protagonist should pick a determinization strategy $\stratdet^F$, which is defined solely by the non-determinism from the initial state, $\stratdet^F(q_L^0, x) = q_L^n$.
  Every choice of $q_L^n$ results in a finite outcome $x^{n+1}$.
  The protagonist thus fails to match the infinite outcome trace $x x x \dots$, so this game does not properly reflect input-failure refinement.

  Remark that this discrepancy for image-infinite IA is not caused by the division of actions into inputs and outputs.
  After all, the IA in Figure~\ref{fig:problems-ati-image-finite} contain only output transitions.
  A similar game-characterization for ordinary trace inclusion would thus also require image-finite models.

  \begin{figure}
    \tikzlts
    \begin{center}
      \begin{tikzpicture}[node distance=16mm]
	\node [initial,initial text=,named] (0) {$q_K^0$};
	\node (init-text) [above left=-2mm and 0 of 0] {$s_K$};
	\path
	(0) edge [loop above] node {$x$} (1)
  (0) edge [loop below] node {$a$} (1)
	;
      
  \node [align=center] [right=8mm of 0] {$\ioco$\\$\uioco$\\$\ir$\\$\iuoe$\\$\notrel\aaee$};
      
	\node [named,initial,initial text=] [right=28mm of 0] (0) {$q_L^0$};
	\node (init-text) [above left=0mm and 0 of 0] {$s_L$};
	\node [named] (2) [above right=10mm and 10mm of 0] {$q_L^1$};
	\node [state] (2a) [right of=2] {};
	\node [named] (3) [above right=0mm and 10mm of 0] {$q_L^2$};
	\node [state] (3a) [right of=3] {};
	\node [state] (3b) [right of=3a] {};
	\node [named] (4) [above right=-10mm and 10mm of 0] {$q_L^3$};
	\node [state] (4a) [right of=4] {};
	\node [state] (4b) [right of=4a] {};
	\node [state] (4c) [right of=4b] {};
	\node [] (5) [above right=-22mm and 10mm of 0] {};
	\path
	(0) edge node [right=1mm] {$x$} (2)
	(0) edge node [below right=0mm] {$x$} (3)
	(0) edge node [below=0mm] {$x$} (4)
	(0) edge [dotted] node [below=2mm] {$x$} (5)
	(2) edge node [above] {$x$} (2a)
	(3) edge node [above] {$x$} (3a)
	(3a) edge node [above] {$x$} (3b)
	(4) edge node [above] {$x$} (4a)
	(4a) edge node [above] {$x$} (4b)
	(4b) edge node [above] {$x$} (4c)
	;
      \end{tikzpicture}
    \end{center}
    \caption{IA $s_K$ and $s_L$.}
    \label{fig:problems-ati-image-finite}
  \end{figure}
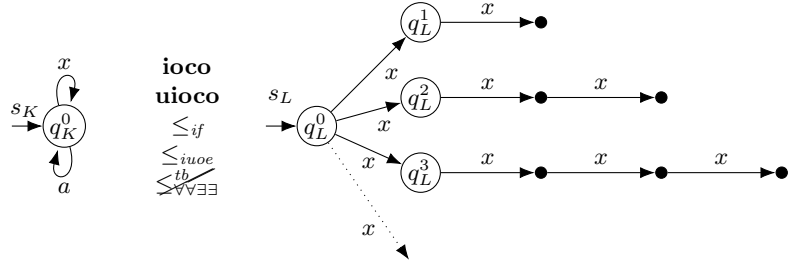
\end{example}
\end{longversion}

\begin{longversion}
We have now presented two games: alternating-trace containment, and a game-characterization of input-failure refinement.
For completeness, we establish that input-failure refinement is indeed weaker.
\end{longversion}

\begin{theorem}
  \label{the:ati-implies-aaee}
  $s_1 \ati s_2 \implies s_1 \aaee s_2$.
\end{theorem}

\begin{proof}
  Assume $s_1 \ati s_2$ (1).
  Changing the order of quantifiers yields a stronger relation: in general, $\exists A, \forall B : \phi(A,B)$ clearly implies $\forall B, \exists A : \phi(A,B)$, for any $A$, $B$ and predicate $\phi$.
  Consequently, (1) implies
  \begin{align*}
    &\forall \stratenvati^2 \in \Stratenvati(s_2), \forall \stratdet^1 \in \Stratdet(s_1), \forall \stratsys^1 \in \Stratsys(s_1), \forall \stratrc^1 \in \Stratrc(s_1),\\
    &\exists \stratenvati^1 \in \Stratenvati(s_1), \exists \stratdet^2 \in \Stratdet(s_2), \exists \stratsys^2 \in \Stratsys(s_2), \exists \stratrc^2 \in \Stratrc(s_2),\\
    &\; \trace(\Outcati[1]) = \trace(\Outcati[2]).
  \end{align*}
  
  Now, we prove that the game played with unrestricted input strategies is stronger than the game played with only trace-based input strategies.
  To prove $s_1 \aaee s_2$, we assume arbitrary $\stratenv^2 \in \Stratenv(s_2)$, $\stratsys^1 \in \Stratsys(s_1)$, $\stratdet^1 \in \Stratdet(s_1)$ and $\stratrc^1 \in \Stratrc(s_1)$. 
  
  Assumption (1) implies that $\stratenvati^1 \in \Stratenvati(s_1)$, $\stratsys^2 \in \Stratsys(s_2)$, $\stratdet^2 \in \Stratdet(s_2)$ and $\stratrc^2 \in \Stratrc(s_2)$ exist such that $\trace(\Outcati[1]) =  \trace(\Outc[2])$ (2).
  We construct $\stratenv^1$ from $\stratenvati^1$ as follows:
  \[
    \stratenv^1(\pi) = 
      \begin{cases}
        a & \text{if $\trace(\pi) a$ is a prefix of $\trace(\Outcati[1])$}\\
          & \quad\text{and $a \in \inp(\after{s_1}{\trace(\pi)})$}\\
        \bot & \text{otherwise}
      \end{cases}
  \]
  Clearly, $\stratenv^1$ is a trace-based input strategy.
  By construction, this strategy also has $\trace(\Outcati[1]) = \trace(\Outc[1])$ (3).
  Moreover, $\trace(\Outc[1]) = \trace(\Outc[2])$ follows from assumptions (2) and (3), which shows that $\stratenv^1$, $\stratsys^2$, $\stratdet^2$ and $\stratrc^2$ are witnesses proving $s_1 \aaee s_2$.
  \qed
\end{proof}

\section{Alternating Simulation}
\label{sec:alternating-simulation}

In \cite{Aluetal98}, two alternating refinement relations have been introduced for alternating transition systems: alternating-trace containment
and alternating simulation. As shown in \cite{Aluetal98}, alternating simulation is stronger than alternating-trace containment, and
both relations coincide for deterministic alternating transition systems.
This should also hold for instantiations on interface automata.
We thus compare our adaptation of alternating-trace containment to the adaptation of alternating simulation from~\cite{AlHe01}.

\begin{definition}
  \label{def:as}
  {\normalfont\cite{AlHe01}}
  Let $s_1$, $s_2 \in \lts$.
  Then $R \subseteq Q_1 \times Q_2$ is an \emph{alternating simulation from $s_1$ to $s_2$} if for all $(q_1, q_2) \in R$, 
  \begin{itemize}
    \item 
      $\out(q_1) \subseteq \out(q_2)$ and $\inp(q_2) \subseteq \inp(q_1)$, and
    \item
      for all $\ell \in \out(q_1) \cup \inp(q_2)$ and $q_1' \in (\after{q_1}{\ell})$, there is a $q_2' \in (\after{q_2}{\ell})$ such that $q_1' \mathrel{R} q_2'$.
  \end{itemize}
  The greatest alternating simulation is denoted $\as$.
  We write $s_1 \as s_2$ to denote $q_1^0 \as q_2^0$.
\end{definition}

Theorem~\ref{the:as-implies-ati} relates the two alternating relations.
Example~\ref{exa:as-implies-ati-strict} proves strictness.

\begin{theorem}
  \label{the:as-implies-ati}
  Let $s_1, s_2 \in \lts$.
  Then $s_1 \as s_2 \implies s_1 \ati s_2$.
\end{theorem}

\begin{proof}
  Assume $s_1 \as s_2$ (1).
  To prove $s_1 \ati s_2$, assume arbitrary strategies $\stratsys^1 \in \Stratsys(s_1)$, $\stratdet^1 \in \Stratdet(s_1)$ and $\stratrc^1 \in \Stratrc(s_1)$ (chosen by the antagonist).
  
  First, let us extend the definition of alternating simulation $\as$ to relate paths $\as$ of $s_1$ and $s_2$, instead of only states, as follows: $\pi_1 \as \pi_2$ holds if $\trace(\pi_1) = \trace(\pi_2)$ and if all pairs of states $(q_1,q_2)$ in $\pi_1$ and $\pi_2$ have $q_1 \as q_2$.
  
  Now, we define a partial function $g : \Path(s_1) \rightharpoonup \Path(s_2)$ inductively as follows:
  \begin{align*}
    g(q_1^0) &= q_2^0\\
    g(\pi_1 \ell q_1) &=
			\begin{cases}
			  g(\pi_1) \ell q_2
					& \text{for some $q_2$ such that $q_1 \as q_2$, }\\
 					& \text{if $g(\pi_1) \downarrow$ and $\ell \in \out(\last(\pi_1)) \cup \inp(\last(g(\pi_1)))$}\\
			  \bot & \text{otherwise}
			\end{cases}
  \end{align*}
  Let $P_1$ be the subset of $\Path(s_1)$ for which $g$ is defined, then clearly $P_1$ is prefix closed.
  We first show that for every $\pi_1 \in \Path(s_1)$,
  \begin{itemize}
    \item
			if $\pi_1 = \pi_1' \ell q_1$, $\pi_1' \in P_1$ and $\ell \in \out(q_1) \cup \inp(q_2)$, then indeed there exists some $q_2$ with $q_1 \as q_2$, and
    \item
			$\pi_1 \as g(\pi_1)$ holds (1).
  \end{itemize}
  We do this by induction on the length of $\trace(\pi_1)$.
	For the base case, let $\pi = q_1^0$, then the former point is vacuously true, and the latter point follows directly from $s_1 \as s_2$.
	For the inductive step, assume that $\pi_1 = \pi_1' \ell q_1$, and assume as induction hypothesis that $\pi_1' \as g(\pi_1')$.
	Then both points follows directly from $\last(\pi_1') \as \last(g(\pi_1'))$ and the construction of $g$.
	
  From the latter point, it also follows that $g$ is well-defined.
  Additionally, $g$ is also clearly injective, so there is a partial inverse $g^{-1}$.
  Let $P_2$ be the subset of $\Path(s_2)$ for which $g^{-1}$ is defined.
  
  Now, we construct $\stratsys^2 \in \Stratsys(s_2)$, $\stratdet^2 \in \Stratdet(s_2)$ and $\stratrc^2 \in \Stratrc(s_2)$ as follows:
  \begin{align*}
    \stratsys^2(\pi_2) &=
      \begin{cases}
        \stratsys^1(g^{-1}(\pi_1)) & \text{if $\pi_2 \in P_2$}\\
        \bot & \text{otherwise}
      \end{cases}\\
    \stratdet^2(\pi_2, \ell) &=
      \begin{cases}
        g(\stratdet^1(g^{-1}(\pi_2))) & \text{if $\pi_2 \in P_2$ and $\stratdet^{-1}(g^{-1}(\pi_2)) \downarrow$}\\
        \text{arbitrary} & \text{otherwise}
      \end{cases}\\
    \stratrc^2(\pi_2) &=
      \begin{cases}
        \stratrc^1(g^{-1}(\pi_1)) & \text{if $\pi_2 \in P_2$}\\
        \text{arbitrary} & \text{otherwise}
      \end{cases}
  \end{align*}

	Furthermore, for arbitrary $\stratenvati^2 \in \Stratenvati(s_2)$, we construct $\stratenvati^1$ as
	\begin{align*}
	  \stratenvati^1(\pi_1) =
	  \begin{cases}
	    \stratenvati^2(\pi_2) & \text{if $\pi_1 \in P_1$}\\
	    \bot & \text{otherwise}
	  \end{cases}
	\end{align*}

  From assumption (1) that $\pi_1 \as g(\pi_1)$ and the construction of $\stratsys^2$, $\stratrc^2$, $\stratdet^2$ and $\stratenvati^1$, it follows that $\Outcati[1] \as \Outcati[2]$.
  This implies that $\trace(\Outcati[1]) = \trace(\Outcati[2])$, so these strategies are witnesses proving $s_1 \ati s_2$.
  \qed
\end{proof}

\begin{example}
  \label{exa:as-implies-ati-strict}
  Readers familiar with ordinary trace containment and simulation will recognize IA $s_M$ and $s_N$ in Figure~\ref{fig:as-implies-ati-strict} as a standard example that shows the difference between the two relations.
  These IA also show the difference between the alternating refinement relations, since $s_M \ati s_N$ and $s_M \not\as s_N$ hold.
  
  Let us first establish $s_M \not\as s_N$.
  Any alternating simulation relation $R$ from $s_M$ to $s_N$ must have $q_M^0 \mathrel{R} q_N^0$.
  Since $q_M^0$ and $q_N^0$ share output $x$, Definition~\ref{def:as} states that either $q_M^1 \mathrel{R} q_N^1$ or $q_M^1 \mathrel{R} q_N^{1\prime}$ should hold.
  If $q_M^1 \mathrel{R} q_N^1$ holds, then $R$ is not an alternating simulation, since $\out(q_M^1) = \{x,y\} \not\subseteq \out(q_N^1) = \{x\}$.
  Likewise, $\out(q_M^1) \not\subseteq \out(q_N^{1\prime}) = \{y\}$ holds, so this disproves alternating simulation.
  
  Clearly, $s_M \ati s_N$ holds: any output strategy by the antagonist in $s_M$ yields a unique path through $s_M$.
  The protagonist can choose output and determinization strategies in $s_N$ such that the outcome in $s_N$ has the same trace.
  
  \begin{figure}
    \tikzlts
    \begin{center}
      \begin{tikzpicture}[node distance=16mm]
        \node [named,initial,initial text=] (0) {$q_M^0$};
        \node (init-text) [above left=0mm and 0 of 0] {$s_M$};
        \node [named] (1) [right=9mm of 0] {$q_M^1$};
        \node [named] (2a) [above right=3mm and 9mm of 1] {$q_M^2$};
        \node [named] (2b) [below right=3mm and 9mm of 1] {$q_M^{2\prime}$};
        \path
        (0) edge node [above] {$x$} (1)
        (1) edge node [above left] {$x$} (2a)
        (1) edge node [below left] {$y$} (2b)
        ;
      
        \node [align=center] [right=30mm of 0] {$\ati$\\$\notrel\as$};
      
        \node [named,initial,initial text=] [right=45mm of 0] (0) {$q_N^0$};
        \node (init-text) [above left=0mm and 0 of 0] {$s_N$};
        \node [named] (1a) [above right=3mm and 9mm of 0] {$q_N^1$};
        \node [named] (1b) [below right=3mm and 9mm of 0] {$q_N^{1\prime}$};
        \node [named] (2a) [right=9mm of 1a] {$q_N^2$};
        \node [named] (2b) [right=9mm of 1b] {$q_N^{2\prime}$};
        \path
        (0) edge node [above left] {$x$} (1a)
        (0) edge node [below left] {$x$} (1b)
        (1a) edge node [below] {$x$} (2a)
        (1b) edge node [above] {$y$} (2b)
        ;
      \end{tikzpicture}
    \end{center}
    \caption{IA $s_M$ and $s_N$.}
    \label{fig:as-implies-ati-strict}
  \end{figure}
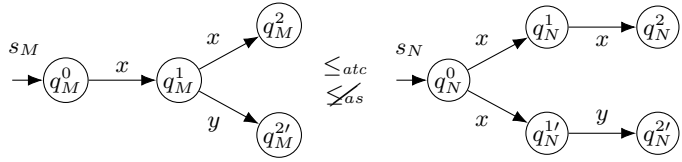
\end{example}

In the deterministic setting, all presented relations coincide.
Clearly, $\ioco$ and $\uioco$ coincide by their definitions.
We show that the remaining relations coincide by proving that the weakest relation implies the strongest.

\begin{theorem}
  \label{the:ir-is-as-det}
  Let $s_1, s_2 \in \lts$, such that $s_2$ is deterministic.
  Then
  \[s_1 \iuoe s_2 \implies s_1 \as s_2.\]
\end{theorem}

\begin{proof}
  Assume $s_1 \iuoe s_2$ (1).
  We show $s_1 \as s_2$ by proving that the relation
  \begin{eqnarray*}
    R & = & \{ (q_1, q_2) \mid \exists \sigma \in \oetraces(s_1) \cap \iutraces(s_2): \\
      &   & \hspace{2cm} q_1 \in (\after{s_1}{\sigma})\mbox{ and }  q_2 \in (\after{s_2}{\sigma})\}
  \end{eqnarray*}
  is an alternating simulation from $s_1$ to $s_2$.
  
  Since $\epsilon \in \oetraces(s_1) \cap \iutraces(s_2)$, $q_1^0 \in (\after{s_1}{\epsilon})$  and  $q_2^0 \in (\after{s_2}{\epsilon})$, we indeed have
  $(q_1^0, q_2^0) \in R$, as required.
  
  Suppose $(q_1, q_2) \in R$.  Then there exists
  $\sigma \in \oetraces(s_1) \cap \iutraces(s_2)$ (2) such that $q_1 \in (\after{s_1}{\sigma})$ (3) and $q_2 \in (\after{s_2}{\sigma})$ (4).
  
  First, we show $\inp(q_2) \subseteq \inp(q_1)$ by proving that $a \in \inp(q_2)$ implies $a \in \inp(q_1)$:
  \begin{align*}
            & a \in \inp(q_2) \\
    \implies& \forall q_2 \in (\after{s_2}{\sigma}): a \in \inp(q_2) \ttag{assumption (4) and $s_2$ deterministic}\\
    \implies& \sigma a \in \oetraces(s_1) \cap \iutraces(s_2) \ttag{assumption (2) and Definition~\ref{def:input-universal}}\\
    \implies& \sigma a \in \iutraces(s_1) \ttag{assumption (1)}\\
    \implies& a \in \inp(q_1) \ttag{assumption (3) and Definition~\ref{def:input-universal} of $\iutraces$}
  \end{align*}
  Next, we show $\out(q_1) \subseteq \out(q_2)$ by proving $x \in \out(q_1)$ implies  $x \in \out(q_2)$:
  \begin{align*}
            & x \in \out(q_1) \\
    \implies& \exists q_1 \in (\after{s_1}{\sigma}): x \in \out(q_1) \ttag{assumption (3)}\\
    \implies& \sigma x \in \oetraces(s_1) \cap \iutraces(s_2) \ttag{assumption (2) and Definition~\ref{def:output-existential}}\\
    \implies& \sigma x \in \oetraces(s_2) \ttag{assumption (1)}\\
    \implies& x \in \out(q_2) \ttag{assumption (4) and $s_2$ deterministic}
  \end{align*}
  Finally, we prove the transfer condition in Definition~\ref{def:as}.
  Suppose that $\ell \in \out(q_1) \cup \inp(q_2)$ and $q_1' \in (\after{q_1}{\ell})$.
  Now if $\ell \in \out(q_1)$ then, as we established above, $\sigma l \in \oetraces(s_1) \cap \iutraces(s_2)$ and $l \in \out(q_2)$.
  Moreover, if $\ell \in \inp(q_2)$ then, as established above, $\sigma l \in \oetraces(s_1) \cap \iutraces(s_2)$ and $l \in \out(q_2)$.
  Using assumptions (3) and (4), this implies that there exists a $q_2' \in (\after{q_2}{\ell})$
  such that $(q_1', q_2') \in R$, as required.
  
  This proves that $R$ is an alternating simulation from $s_1$ to $s_2$, and therefore $s_1 \as s_2$.
  \qed
\end{proof}

Efficient algorithms for checking alternating simulation exist~\cite{AlHe01}.
Since all relations treated in this paper coincide if the right-hand IA is deterministic, an approach to decide any of these relations between two IA could be to transform the right-hand IA to a deterministic IA, preserving that relation, and then use the algorithm for alternating simulation.
The standard subset-construction for determinization \cite{HU79,Tre96}, however, does not preserve input-failure refinement, as Example~\ref{exa:wrong-determinization} shows.
We recall the subset-construction in Definition~\ref{def:subset-construction}.

\begin{definition}
  \label{def:subset-construction}
  Let $s \in \lts$.
  Then $\det(s) = (\mathcal{P}(Q_s) \setminus \{\emptyset\}, I, O, T_{\det}, \{q_s^0\})$, with
  \[
    T_{\det} \isdef \{(Q,\ell,\after[s]{Q}{\ell}) \mid Q \subseteq Q_s, \ell \in L,
    (\after[s]{Q}{\ell}) \neq \emptyset\}.
  \]
\end{definition}

\begin{example}
  \label{exa:wrong-determinization}
  
  Consider the interface automaton $s_A$ from Figure~\ref{fig:ir}.
  We perform the subset construction on $s_A$, and obtain $\det(s_A)$ as shown in Figure~\ref{fig:determinization}.
  Whereas $s_A$ contains failure trace $a\overline{a}$, $\det(s_A)$ does not.
  As a consequence, $s_A \not\ir \det(s_A)$, so the model is changed with respect to input-failure refinement.
  
  \begin{figure}
    \tikzlts
    \begin{center}
      \begin{tikzpicture}[node distance=19mm]
	\node [namedlong,initial,initial text=] (0) {$\{q_A^0\}$};
	\node (init-text) [above left=0mm and -6mm of 0] {$\det(s_A)$};
	\node [namedlong] (12) [right=5mm of 0] {$\{q_A^1,q_A^2\}$};
	\node [namedlong] (1) [right=5mm of 12] {$\{q_A^1\}$};
	\node [namedlong] (2) [right=5mm of 1] {$\{q_A^2\}$};
	\path
	(0) edge [loop below] node {$x$} (0)
	(0) edge node [below] {$a$} (12)
	(12) edge node [below] {$x$} (1)
	(12) edge [bend left] node [above] {$a$} (2)
	(1) edge [loop below] node {$x$} (1)
	(1) edge node [below] {$a$} (2)
	;
	
	\node [namedlong,initial,initial text=,right=53mm of 0] (0) {$\{q_A^0\}$};
	\node (init-text) [above left=0mm and -6mm of 0] {$\detiu(s_A)$};
	\node [namedlong] (12) [right=5mm of 0] {$\{q_A^1,q_A^2\}$};
	\node [namedlong] (1) [right=5mm of 12] {$\{q_A^1\}$};
	\node [namedlong] (2) [right=5mm of 1] {$\{q_A^2\}$};
	\path
	(0) edge [loop below] node {$x$} (0)
	(0) edge node [below] {$a$} (12)
	(12) edge node [below] {$x$} (1)
	(1) edge [loop below] node {$x$} (1)
	(1) edge node [below] {$a$} (2)
	;
      \end{tikzpicture}
    \end{center}
    \caption{The subset-construction (standard and input-universal, respectively) performed on $s_A$.
    Only the part reachable from the initial state is shown.}
    \label{fig:determinization}
  \end{figure}
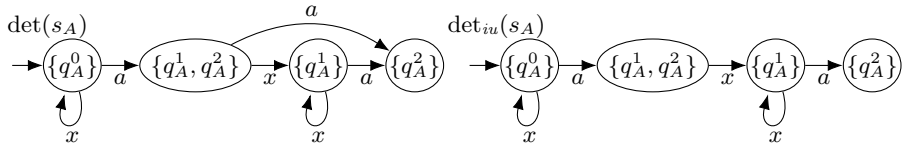
\end{example}

\begin{shortversion}
We introduce a determinization variant which respects input-universality, in order to preserve input-failure refinement, instead of normal trace containment.  
\end{shortversion}
\begin{longversion}
We introduce a determinization variant which respects input-universality, in order to preserve input-failure refinement.
Note that this does not preserve other relations, such as traditional trace containment.  
\end{longversion}

\begin{definition}
  \label{def:determinization}
  Let $s \in \lts$.
  Then $\detiu(s) = (\mathcal{P}(Q_s) \setminus \emptyset, I, O, T_{\detiu}, \{q_s^0\})$, with
  \begin{align*}
    T_{\detiu} &\isdef \{(Q,\ell,\after[s]{Q}{\ell}) \mid \emptyset \neq Q \subseteq Q_s, 
    \ell \in \inp_s(Q) \cup \out_s(Q)\}.
  \end{align*}
\end{definition}

\begin{example}
  Figure~\ref{fig:determinization} also shows the input-universal determinization $\detiu(s_A)$ of $s_A$.
  Since $a \not\in \inp(q_A^2)$, input $a$ is not universally enabled in $\{q_A^1,q_A^2\}$, which implies that this state has no $a$-transition in $\detiu(s_A)$.
  In fact, the reader may check that $s_A \ireq \detiu(s_A)$ holds.
  This also follows from Theorem~\ref{tm:determinization}.
\end{example}

\begin{longversion}
\begin{fact}
  \label{fac:determinization-deterministic}
  For $s \in \lts$, $\detiu(s)$ is deterministic.
\end{fact}

\begin{lemma}
  \label{lem:determinization-after}
  Let $s \in \lts$ and $\sigma \in \traces(\detiu(s))$.
  Then $\sigma \in \traces(s)$ and
  \[(\after{\detiu(s)}{\sigma}) = \{\after{s}{\sigma}\}\]
\end{lemma}

\begin{proof}
  By induction on the length of $\sigma$.
  The base case $\sigma = \epsilon$ is trivial, so assume $\sigma = \sigma'\ell$, and assume as induction hypothesis that $(\after{\detiu(s)}{\sigma'}) = \{\after{s}{\sigma'}\}$ (IH1) and that $\sigma' \in \traces(s)$, so $|\after{s}{\sigma'}| \ge 1$ (IH2).
  
  Then
  \begin{align*}
             & \sigma \in \traces(\detiu(s))\\
    \implies & \ell \in \inp_{\detiu(s)}(\after{\detiu(s)}{\sigma'}) \cup \out_{\detiu(s)}(\after{\detiu(s)}{\sigma'})
	  \ttag{determinism of $\detiu(s)$}\\
    \implies &\ell \in \inp_{\detiu(s)}(\{\after{s}{\sigma'}\}) \cup \out_{\detiu(s)}(\{\after{s}{\sigma'}\})
	  \ttag{(IH1)}\\
    \implies &\ell \in \inp_s(\after{s}{\sigma'}) \cup \out_s(\after{s}{\sigma'}) \quad (1)
	  \ttag{Construction of $T_{\detiu}$ in Definition~\ref{def:determinization}}\\
    \implies &\sigma\ell \in \traces(s)
	  \ttag{(IH2)}
  \end{align*}
  Furthermore, we have
  \begin{align*}
     & \after{\detiu(s)}{\sigma} \\
    =& \after[\detiu(s)]{(\after{\detiu(s)}{\sigma'})}{\ell}\\
    =& \after[\detiu(s)]{\{\after{s}{\sigma'}\}}{\ell} \ttag{(IH1)}\\
    =& \{\after[s]{\{\after{s}{\sigma'}\}}{\ell}\} \ttag{Construction of $T_{\detiu}$ in Definition~\ref{def:determinization} and (1)}\\
    =& \{\after{s}{\sigma}\} \qedtag
  \end{align*}
\end{proof}

\begin{proposition}
  \label{pro:det-preserves-rcl}
  Let $s \in \lts$.
  Then $\rcl(\irtraces(s)) = \rcl(\irtraces(\detiu(s)))$
\end{proposition}

\begin{proof}
  $(\subseteq)$
  We first prove $\irtraces(s) \subseteq \rcl(\irtraces(\detiu(s)))$, by showing that any trace $\sigma \in \irtraces(s)$ (1) is also in $\rcl(\irtraces(\detiu(s)))$.
  We distinguish two cases, based on (1) and the form of $\irtraces$ in Definition~\ref{def:ir-inclusion}:
  \begin{itemize}
    \item 
      Consider $\sigma \in \traces(s)$ (2).
      If $\sigma = \epsilon$ or $\sigma \in \traces(\detiu(s))$, then proving $\sigma \in \rcl(\irtraces(\detiu(s)))$ is trivial, so assume $\sigma \neq \epsilon$ and $\sigma \not\in \traces(\detiu(s))$ (3).
      Then $\sigma = \rho\ell\tau$ (4) for some $\rho \in \traces(\detiu(s))$ (5) and $\rho\ell \not\in \traces(\detiu(s))$ (6).
      As such, there is some $Q \subseteq Q_s$ such that $\{q_s^0\} \xrightarrow{\sigma'}_{\detiu(s)} Q$ (7).
  
      Then $\ell \in I \cup O$ by (4), and $Q \,\notrel{\xrightarrow{\ell}}_{\detiu(s)}$ by (6) and (7).
      Consequently, $\ell \not \in \inp_s(Q) \cup \out_s(Q)$ (8) by construction of $T_{\detiu}$.
      Furthermore, if $\ell \in O$, then $\ell \in \out_s(Q)$ by (2), which would contradict (8), so $\ell \in I$.
      Thus, (5), (6) and (8) imply $\rho \overline{a} \in \irtraces(\detiu(s))$, so $\sigma \in \rcl(\irtraces(s))$.
    \item
      If $\sigma = \sigma'\overline{a}$, then $\sigma' \in \traces(s)$ and $a \not\in \inp_s(\after{s}{\sigma'})$.
      Then this also implies $a \not\in \inp_{\detiu(s)}(\{\after{s}{\sigma'}\})$ by construction of $T_{\detiu}$, and furthermore $a \not\in \inp_{\detiu(s)}(\after{\detiu(s)}{\sigma})$ by Lemma~\ref{lem:determinization-after}, so $\sigma'\overline{a} \in \rcl(\irtraces(\detiu(s)))$ holds.
  \end{itemize}
  This proves $\irtraces(s) \subseteq \rcl(\irtraces(\detiu(s)))$.
  Thus, $\rcl(\irtraces(\detiu(s)))$ is an input-failure closed superset of $\irtraces(s)$.
  Then it must be larger than the smallest input-failure closed supserset of $\irtraces(s)$, that is, $\rcl(\irtraces(s)) \subseteq \rcl(\irtraces(\detiu(s)))$.
  
  $(\supseteq)$
  Let $\sigma \in \rcl(\irtraces(\detiu(s)))$ (1).
  Then we prove $\sigma \in \rcl(\irtraces(s))$.
  We distinguish two cases based on the form of $\rcl(\irtraces(\detiu(s)))$ in Definition~\ref{def:ir-inclusion}.
  \begin{itemize}
    \item
      If $\sigma = \sigma' \overline{a} \in \irtraces(\detiu(s))$, or if $\sigma$ has some prefix $\sigma' a$ with $\sigma \overline{a} \in \irtraces(\detiu(s))$, then this implies that $(\after{\detiu(s)}{\sigma'}) = \{Q\}$ (2) for some $Q \subseteq Q_s$ with $a \not \in \inp_{\detiu(s)}(\{Q\})$ (3).
      Then also $(\after{s}{\sigma'}) = Q$ by (2) and Lemma~\ref{lem:determinization-after}, and $a \not \in \inp_s(Q)$ by construction of $T_{\detiu}$ in Definition~\ref{def:determinization}.
      Consequently, $\sigma'\overline{a} \in \irtraces(s)$ holds, so also $\sigma \in \rcl(\irtraces(s))$ holds.
    \item
      If $\sigma \in \traces(\detiu(s))$, then $\sigma \in \traces(s)$ by Lemma~\ref{lem:determinization-after}, so $\sigma \in \rcl(\irtraces(s))$ holds.
      \qed
  \end{itemize}  
\end{proof}
\end{longversion}

\begin{theorem}
\label{tm:determinization}
  Let $s \in \lts$.
  Then $s \ireq \detiu(s)$.
\end{theorem}

\begin{proof}
  Follows directly from Propositions~\ref{pro:rcl-irtraces-canonical} and~\ref{pro:det-preserves-rcl}.
  \qed
\end{proof}

\begin{corollary}
  \label{cor:det-preserves-relations}
  Let $s_1, s_2 \in \lts$.
  Then
  \begin{align*}
    s_1 \as \detiu(s_2) &\iff s_1 \ir s_2 \\
    s_1 \as \detiu(s_2) &\iff s_1 \iuoe s_2\\
    \Delta(s_1) \as \detiu(\Delta(s_2)) &\iff s_1 \uioco s_2\\
    \mbox{if } s_2 \mbox{ image finite then } s_1 \as \detiu(s_2) &\iff s_1 \aaee s_2
  \end{align*}
\end{corollary}

Completing the lattice in Figure~\ref{fig:overview}, one may expect alternating simulation to be the strongest relation, in the same way that ordinary simulation is the strongest in the spectrum of Van Glabbeek~\cite{vG01}.
But Example~\ref{exa:ioco-and-as} shows that alternating simulation is neither stronger nor weaker than $\ioco$.
This supports the conclusion in~\cite{JaTr19} that it is hard to $\ioco$-implement a given specification: even an alternating simulation refining implementation may not be $\ioco$-conformant.

\begin{example}
  \label{exa:ioco-and-as}
  Consider IA $s_P$ and $s_Q$ in Figure~\ref{fig:ioco-and-as}.
  They have $\Delta(s_P) \as \Delta(s_Q)$, as shown by the alternating simulation relation $\{(q_P^0,q_Q^0), (q_P^1,q_Q^1), (q_P^2,q_Q^2)\}$.
  However, $s_P \notrel{\ioco} s_Q$ since $\out(\after{s_P}{ab}) = \{x\} \not\subseteq \out(\after{s_Q}{ab}) = \{y\}$.
  
  Vice versa, in Figure~\ref{fig:as-implies-ati-strict}, $s_M \ioco s_N$ clearly holds, whereas $s_M \not\as s_N$.
  \begin{figure}
    \begin{center}
      \begin{tikzpicture}[node distance=19mm]
      {\tikzlts
      \node [named,initial,initial text=] (i0) {$q_P^0$};
      \node (init-text) [above left=-2mm and 0mm of i0] {$s_P$};
      \node [named] (i1) [right of=i0] {$q_P^1$};
      \node [named] (i2) [above right=4.5mm and 10mm of i1] {$q_P^2$};
      \node [named] (i2p) [below right=4.5mm and 10mm of i1] {$q_P^{2\prime}$};
      \path
      (i0) edge [loop below] node {$b$} (i0)
      (i0) edge [loop above, densely dashed] node {$\delta$} (i0)
      (i0) edge node [above] {$a$} (i1)
      (i1) edge node [above left] {$a$} (i2)
      (i1) edge [loop above, densely dashed] node {$\delta$} (i1)
      (i1) edge node [below left] {$b$} (i2p)
      (i2) edge [in=150,out=120,loop] node [left] {$a$} (i2)
      (i2) edge [in=30,out=60,loop] node [right=0.5mm] {$x$} (i2)
      (i2) edge [in=-60,out=-30,loop] node [right=0.5mm] {$b$} (i2)
      (i2p) edge [in=30,out=60,loop] node [above right] {$a$} (i2p)
      (i2p) edge [in=210,out=240,loop] node [left=0.5mm] {$b$} (i2p)
      (i2p) edge [in=-60,out=-30,loop] node [right] {$x$} (i2p)
      ;
      
      \node [align=center, below right=-3mm and 44mm of i0] {$\as$\\\vspace{-3mm}\\$\notrel{\ioco}$};
      
      \node [named,initial,initial text=] (0) [right=61mm of i0] {$q_Q^0$};
      \node (init-text) [above left=-2mm and 0mm of 0] {$s_Q$};
      \node [named] (1) [above right=7mm and 10mm of 0] {$q_Q^1$};
      \node [named] (1p) [below right=7mm and 10mm of 0] {$q_Q^{1\prime}$};
      \node [named] (2) [right of=1] {$q_Q^2$};
      \node [named] (2p) [right of=1p] {$q_Q^{2\prime}$};
      \path
      (0) edge [loop below, densely dashed] node {$\delta$} (0)
      (0) edge node [above left] {$a$} (1)
      (0) edge node [below left] {$a$} (1p)
      (1) edge node [above] {$a$} (2)
      (1) edge [loop below, densely dashed] node {$\delta$} (1)
      (1p) edge [loop above, densely dashed] node {$\delta$} (1p)
      (1p) edge node [below] {$b$} (2p)
      (2) edge [loop right] node {$x$} (2)
      (2p) edge [loop right] node {$y$} (2p)
      ;
      }
      \path[loosely dotted, thick, color=myGreen]
      (i0) edge [out=-49,in=-131,looseness=1.0] (0)
      (i1) edge [out=0,in=190,out looseness=1.2] (1)
      (i2) edge [out=0,in=140,in looseness=0.5,out looseness=0.2] (2)
      ;
      \end{tikzpicture}
    \end{center}
    \caption{
    Alternating simulation is not stronger than $\ioco$.
    The dotted lines indicate states related by alternating simulation.
    }
    \label{fig:ioco-and-as}
  \end{figure}
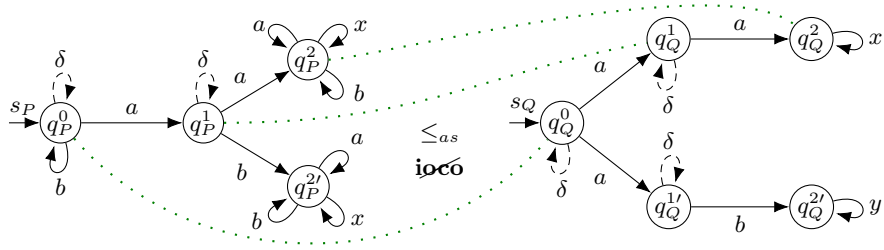
\end{example}
\section{Conclusion and Future Work}

We provided strong links between the $\ioco$ testing theory and alternating refinement theory on interface automata.
The overlap between the relations from these independently developed theories indicate that they express a natural notion of refinement.
Based on the strong correspondence between elements in testing theory and concepts from game theory~\cite{BoSt18}, the provided links pave the way for using results from game theory in testing with $\uioco$ and $\ir$.
We have also shown that alternating-trace containment does not lend itself well to an observational interpretation, but that a slight modification of the game rules solves this.
Likewise, we deem $\ioco$ to be too strong for a practical implementation relation, as alternating simulation is not stronger.

To ease the comparison between $\ioco$ theory and alternating refinements, we introduced two relations which may be of interest in their own right.
Input-failure refinement has a direct connection to alternating simulation, and to $\uioco$ when quiescence is added explicitly. 
Because of its straightforward observational interpretation, input-failure refinement should be suitable in conformance testing.
A next step is to formalize and implement testing algorithms for this relation.
The alternative characterization in terms of input-existential and output-universal traces may serve as a tool in formal reasoning.

More conformance and refinement relations for systems with inputs and outputs exist, e.g.,  in the context of testing theory~\cite{HeTr97,FrTr07} and I/O automata theory~\cite{ReingoldWZ92,SegalaGSL98}.
It would be interesting to include these works in our spectrum.
An additional improvement is to include internal transitions, as commonly found in interface automata and labelled transition systems.

\FloatBarrier

\bibliographystyle{plain}
\begin{shortversion}
\bibliography{bibshort}{}  
\end{shortversion}
\begin{longversion}
\bibliography{bibshort}{}  
\end{longversion}

\end{document}